\documentclass{elsarticle}
\usepackage[caption = false]{subfig}
\usepackage[multiple]{footmisc}
\usepackage{eurosym}
\usepackage{xspace}
\usepackage{xcolor}
\usepackage{caption}

\usepackage[draft]{changes}

\usepackage{hyperref}   

\usepackage{amsmath,amssymb,amsthm,mathrsfs,amsfonts,dsfont}

\usepackage{array}
\usepackage{amsthm}
\usepackage{epstopdf}
\usepackage{amsmath}
\usepackage{setspace}
\usepackage[linesnumbered,ruled,vlined,procnumbered]{algorithm2e}
\usepackage{natbib}
\usepackage{color}
\usepackage{rotating}
\usepackage{tikz}
\usepackage{subfig}
\usepackage[toc,page]{appendix}
\usepackage[font=small,skip=0pt]{caption}
\usepackage{eurosym}
\usepackage{babel}
\usepackage{ulem}
\usepackage{multirow}
\usepackage{float}

\newcommand{\Break}{\textbf{break}}

\newcommand{\pkg}[1]{{\normalfont\fontseries{b}\selectfont #1}}
\let\proglang=\textsf

\newcommand*\colvec[3][]{
    \begin{pmatrix}\ifx\relax#1\relax\else#1\\\fi#2\\#3\end{pmatrix}
}

\newcommand*\colvecfive[5][]{
    \begin{pmatrix}\ifx\relax#1\relax\else#1\\\fi#2\\#3 \\#4\\#5\end{pmatrix}
}

\newcommand{\persistence}{{{\cal P}^{\star}_{\cal C}}\xspace}

\theoremstyle{remark}

\newtheoremstyle{myproposition} 
  {3pt}                          
  {3pt}                          
  {\itshape}                     
  {}                             
  {\bfseries}                    
  {.}                            
  { }                            
  {}                             

\theoremstyle{myproposition}
\newtheorem{proposition}{Proposition}

\newtheorem{definition}{Definition}
\theoremstyle{definition}

\numberwithin{table}{section}
\numberwithin{figure}{section}

\graphicspath{{figures/}}  

\journal{}

\begin{document}

\begin{frontmatter}

\title{Null-adjusted persistence function for high-resolution community detection}

\author[unimib]{Alessandro Avellone}
\author[unimi]{Paolo Bartesaghi}
\author[unitn]{Stefano Benati}
\author[unicip]{Christos Charalambous}
\author[unimib]{Rosanna Grassi}

\cortext[cor]{\emph{Corresponding author. email: rosanna.grassi@unimib.it}}

\address[unimib]{University of Milano - Bicocca, Via Bicocca degli Arcimboldi 8, 20126 Milano, Italy}

\address[unimi]{University of Milano, Via Conservatorio 7, 20122 Milano, Italy}

\address[unitn]{University of Trento, Via Verdi 26, 38122 Trento, Italy}

\address[unicip]{University of Cyprus, Department of Economics, PO Box 20537, 1678 Nicosia, Cyprus}

\begin{abstract}
    Modularity and persistence probability are two widely used quality functions for detecting communities in complex networks. In this paper, we introduce a new objective function called null-adjusted persistence, which incorporates features from both modularity and persistence probability, as it implies a comparison of persistence probability with the same null model of modularity. We prove key analytic properties of this new function. We show that the null-adjusted persistence overcomes the limitations of modularity, such as scaling behavior and resolution limits, and the limitation of the persistence probability, which is an increasing function with respect to the cluster size. We propose to find the partition that maximizes the null-adjusted persistence with a variation of the Louvain method and we tested its effectiveness on benchmark and real networks. We found out that maximizing null-adjusted persistence outperforms modularity maximization, as it detects higher resolution partitions in dense and large networks. 
\end{abstract}

\end{frontmatter}

\section{Introduction}

Community detection is an essential research area of network science: it seeks to uncover densely connected subgroups of nodes, interpreted as communities  -- also referred to as clusters or modules --  within a given network. 
Loosely speaking, communities should be characterized by multiple links between their members and easy, fast communication between them. In contrast, they should have sparse connections to external nodes and less accessibility to information outside the groups.

Identifying these structures is critical to understanding the underlying organization and functionality of networks in various domains, including social (\cite{2019Zhang, calderoni2017, benati2024}) 
and biological systems (\cite{2019Zhang, calderoni2017, benati2024, 2002Girvan}), and economic and financial applications (\cite{2008Allen, bartesaghi2020, grassi2021}). 
An essential contribution to community detection that shaped the modern study of the field is contained in (\cite{2002Girvan}). To detect communities, the authors proposed an algorithm based on hierarchical partitioning, in which arcs are progressively deleted depending on their edge betweenness. In the following two decades, the contributions to community detection by scientists from the fields of network theory, physics, computer science, operations research, and inferential statistics have been numerous and diversified, as evidenced by the surveys \cite{2010Fortunato, Fortunato2016,2023Li}.

Among others, the following approach has been a recognizable research trend: first, a network statistic is elaborated to discern groups that can be interpreted as communities from the ones that cannot; then, the network community structure is determined through the maximization of the proposed statistic used as objective function, that is, the community structure emerges as the partition that maximizes that statistic. The most widely used community statistic is modularity (\cite{2006Newman}). Modularity evaluates the difference between the actual number of edges within communities and the expected number of such edges in a random network with the same degree distribution (the so-called \textit{configuration model}, see \cite{Newmann2002}).
Modularity has been maximized through exact and heuristic methods. One of the first and most popular heuristics is the Louvain algorithm (\cite{2008Blondel}), in which nodes and communities are progressively merged until a (local) maximal modularity has been found. 
The exact methods maximize the modularity through a clique partitioning problem, formulated as an Integer Linear Programming (ILP) problem and solved with specific techniques, see \cite{Agarwal2008, liberti2010, dinh2015, Zhu2020}. 
However, maximizing modularity and clique partitioning are NP-hard problems. Only instances of moderate size can be solved, and therefore a large effort has been devoted to developing and improving heuristics, to the point that sometimes the purpose of maximizing the modularity has been lost (readers can refer to \cite{Aref2023} for an accurate list of modularity maximization heuristic algorithms).

Nevertheless, modularity suffers of some drawbacks. Specifically, it is biased by the so-called resolution limit, see (\cite{2007Fortunato,brandes2008,lancichinetti2011,LU2020}). Although not apparent from its definition, the size of the network has an impact on the maximization of modularity. It has been proved that, if communities are small enough with respect to the network size, then they are not recognizable as they are merged into larger groups. To amend this problem, some authors proposed some corrections to modularity, such as the modularity density, \cite{Li2008, Costa2015}, or the $z$-modularity, see \cite{Miyauchi2016}. Other authors suggested imposing some linear constraint to the ILP model of modularity maximization to obtain stronger community definitions, \cite{Cafieri2012}. Other authors, see 
\cite{Temprano2024}, suggested applying a measure, the normalized cuts, previously applied to image segmentation, see \cite{shi2000}.

In \cite{Piccardi2011} a new index, called persistence probability, is proposed and described as the probability that a random walker starting in a given cluster  
will move to a node within the same cluster in the next iteration.
As will be seen, persistence probability is defined as the ratio between the internal edges of a cluster and all the edges adjacent to that cluster; therefore, as is the case for modularity, a large number of internal arcs is an important feature of the community definition. However, the two measures are based on two different arguments. To define a community, modularity emphasizes static bonds \textit{within} its members, while persistence emphasizes the role of dynamic communication \textit{between} its members. A community with a high persistence probability reflects a scenario where information spreading across the network remains within a given community, and it is not shared with the rest of the network. Indeed, it has been used to characterize criminal gangs \cite{calderoni2017}, to analyze the world trade web structure \cite{piccardi2012},  to identify cohesive and persistent communities in dynamic social networks (\cite{2014Nguyen}) and in datasets from platforms like Facebook and Twitter (\cite{2020Tulu}). 
Algorithms and subroutines for the persistence probability are not as well-developed and tested as those for the modularity. The preliminary problem of finding \textit{one} community that maximizes a slightly modified version of persistence probability has been discussed \cite{Avellone24}. It has been found that the problem can be formulated as a fractional integer programming problem, and therefore exact and heuristic methods are proposed and applied to real network data.      

In this paper, we deal with the problem of finding the graph partition with maximum persistence probability. In Section \ref{sec2}, we introduce the definition of persistence probability in the context of Markov chains and formulate the maximization problem. Since persistence probability tends to increase with respect to the size of the cluster, we propose a correction to persistence by defining a new function, which we call \textit{null-adjusted persistence}. Drawing inspiration from modularity, it adjusts the persistence probability with a term representing the expected persistence under the configuration model, that is, a null hypothesis that assumes a random graph with no community structure. In Section \ref{sec3}, we present some analytical results. In particular, we highlight the differences between null-adjusted persistence and modularity, and we show that the former is scale-independent and not affected by the resolution limit. Then, we prove a necessary and sufficient condition for merging two distinct clusters and improving the null-adjusted persistence. In Section \ref{sec4}, we develop heuristic and exact algorithms to solve the proposed problem. We then apply these algorithms to test and compare null-adjusted persistence and modularity on synthetic networks and real data. Our results confirm that, in relevant cases, null-adjusted persistence recognizes network structure better than modularity. Finally, in the conclusion, we discuss possible future developments.

\section{Persistence probability by Markov chains}
\label{sec2}
In this section, we introduce the definition of persistence probability in the context of Markov chains.
Let $G=(V, E, W)$ be a weighted, undirected graph (or network), where $V$ is the set of $n$ vertices (or nodes), $E$ is the set of edges (or arcs), and $W$ is the set of edge weights. Let $n = \lvert V \rvert$ be the cardinality of $V$. The subgraph induced by $V' \subseteq V$ is the graph $G_{V'}$ whose vertex set is $V'$ and whose edge set consists of all the edges in $E$ that have both endpoints in ${V'}$. A \textit{clustering} of $G$ is a partition of the network nodes $\Pi=\{ \mathcal{C}_{1}, \dots, \mathcal{C}_{q}\}$ where $\mathcal{C}_{\alpha} \subseteq V$, $\mathcal{C}_{\alpha}\neq \emptyset$ and the induced subgraph $G_{\mathcal{C}_{\alpha}}$ is connected, for all $\alpha=1,\ldots, q$. We denote the set of all possible clusterings of a graph $G$ with $\wp(G)$. 

The information about the weights assigned to the edges is contained in the $n$-square matrix ${\bf W}=\left[ w_{ij} \right]$, where $w_{ij} \ge 0$ is the weight of the edge between nodes $i$ and $j$. The underlying graph, obtained neglecting the arcs weights, is described by the adjacency matrix ${\bf A}=\left[ a_{ij} \right]$, where $a_{ij}=1$ if $w_{ij}>0$ and $a_{ij}=0$ otherwise. The strength and the degree of a node $i$ are defined, respectively, by $s_{i}=\sum_{j=1}^{n}{w_{ij}}$ and $k_{i}=\sum_{j=1}^{n}{a_{ij}}$. We denote by $\mathbf{s}$ and $\mathbf{k}$ the vectors of the strengths and the degrees, respectively, and by $S=\sum_{i<j}{w_{ij}}$ the total strength of the network. 

Since persistence probability is based on random walks on a graph, we assume that the graph $G$ is connected. We can place a homogeneous discrete-time Markov chain over the network, whose space of states coincides with the set of nodes $V$. Specifically, a discrete-time Markov chain on the network is any stochastic process represented by a sequence of $n-$state vectors $\pi(t)=\left(\pi_{1}(t), \dots ,\pi_{n}(t)\right)$ such that the probability of being in any state at a given step $t$ depends only on the state at the previous step $t-1$. The process is thus described by the equation $\pi(t+1)=\pi(t){\bf P}$, where ${\bf P}=\left[ p_{ij} \right]$ is an $n$-squared matrix and $0\leq p_{ij}\leq 1$ is the conditional probability that a random walker jumps from node $i$ to node $j$ at each step, assuming it is in $i$. ${\bf P}$ is a row-stochastic matrix called transition probability matrix.
The simplest choice of transition matrix is associated with the \textit{natural} Markov chains. It assumes that the probability of transition from a node $i$ is uniformly distributed over its neighbors, while it is zero for nodes not directly connected - i.e. not adjacent - to $i$, so that $p_{ij}=\frac{w_{ij}}{s_{i}}$. 
Notice that $\mathbf{s}$ is entrywise positive, being $G$ connected.
The transition matrix can then be expressed as  ${\bf P}=({\rm diag}\, \mathbf{s})^{-1}{\bf W}$, being  ${\rm diag}\, \mathbf{s}$ the diagonal matrix whose diagonal entries are the elements of the vector $\mathbf{s}$. 
A discrete-time Markov chain is homogeneous if the one-step transition probabilities are invariant with respect to time. This implies that the $k$-step transition probabilities can be computed through the $k$-th power ${\bf P}^k$ of the transition matrix.

We also assume that the network $G$ is non-bipartite, hence the Markov chain over $G$ is ergodic, i.e it is always possible to go from any state to any other state and the chain does not exhibit periodic behavior. As a consequence, matrix ${\bf P}$ is irreducible, and there exists a unique stationary state solution of the eigenvalue equation $\pi({\infty}) = \pi({\infty}) {\bf P}$, of the form $\pi_{i}({\infty})=\frac{s_{i}}{\sum_{i=1}^{N} s_{i}}=\frac{s_{i}}{2S}$. Moreover, the sequence of matrices ${\bf P}^k$ converges to a rank-one matrix, that is ${\bf P}^{\infty}= \lim_{k\to \infty} {\bf P}^{k}$, and $\pi({\infty}) = \pi(0) {\bf P}^{\infty}$, where ${\bf P}^{\infty}$ is the \textit{infinite-time} transition matrix containing the transition probabilities of jumping from node $i$ to node $j$ in infinite number of steps.
The stationary probability flux $\phi_{ij}$ along the edge $(i,j)$ 
is finally defined as the actual probability that the walker jumps from $i$ to $j$ in the stationary state, that is $\phi_{ij}=\pi_{i}({\infty})p_{ij}$.

Now, let ${\mathscr{C}}$ be a clustering of $G$ and let ${\mathcal C}_{\alpha}$ and ${\mathcal{C}}_{\beta}$ two distinct elements in ${\mathscr C}$. Let us consider the global stationary probability flux from ${\mathcal{C}}_{\alpha}$ to ${\mathcal{C}}_{\beta}$: $\Phi_{_{{\mathcal{C}}_{\alpha}{\mathcal{C}}_{\beta}}}=\sum_{i\in {\mathcal{C}}_{\alpha}}\sum_{j\in {\mathcal{C}}_{\beta}}\phi_{ij}$. By using this flux between clusters we can induce an aggregated process on a meta-network whose meta-nodes are clusters. This process is known in the literature as lumped Markov chain (\cite{Piccardi2011}). Since the random walker being in two different nodes at the steady state are incompatible events, the probability that it is in the meta-node ${\mathcal{C}}_{\alpha}$ at the steady state is then given by $\pi_{{\mathcal{C}}_{\alpha}}({\infty})=\sum_{i\in {\mathcal{C}}_{\alpha}}\pi_{i}({\infty})$. Therefore, the transition probability from ${\mathcal{C}}_{\alpha}$ to ${\mathcal{C}}_{\beta}$ is given by the conditional probability 
\begin{equation}
\label{LumpedMarkovChain}
p_{_{{\mathcal{C}}_{\alpha}{\mathcal{C}}_{\beta}}}
=
\frac{\Phi_{_{{\mathcal{C}}_{\alpha}{\mathcal{C}}_{\beta}}}}{\pi_{{\mathcal{C}}_{\alpha}}({\infty})}
=
\frac{\sum_{i\in {\mathcal{C}}_{\alpha}}\sum_{j\in {\mathcal{C}}_{\beta}}\pi_{i}({\infty})p_{ij}}{\sum_{i\in {\mathcal{C}}_{\alpha}}\pi_{i}({\infty})}.
\end{equation}
These values represent the entries of the transition probability matrix of the lumped Markov chain. We denote the diagonal element $p_{_{{\mathcal{C}}_{\alpha}{\mathcal{C}}_{\alpha}}}$ of this matrix as ${\cal P}_{_{{\mathcal{C}}_{\alpha}}}$ and we call it the persistence probability of the cluster ${\mathcal{C}}_{\alpha}$ (\cite{Piccardi2011, 2020Patelli}).
In the case of the natural Markov chain on a weighted undirected network, since $p_{ij}=\frac{w_{ij}}{s_{i}}$ and $\pi_{i}({\infty})=\frac{s_{i}}{2S}$, the persistence probability of the generic cluster ${\mathcal{C}}$ is given by
\begin{equation} \label{prs1}
{\cal P}_{\mathcal{C}}=\frac{\sum_{i,j\in {\mathcal{C}}}w_{ij}}{\sum_{i\in {\mathcal{C}}}s_{i}}.
\end{equation}

Hence, the persistence probability for an undirected weighted network is the ratio between the total weight of the edges inside the community ${\mathcal{C}}$ and the total weight of the edges starting from one of the nodes in  ${\mathcal{C}}$ and ending both inside and outside ${\mathcal{C}}$. Finally,
for the an unweighted graph, the persistence probability reduces to ${\cal P}_{\mathcal{C}}=\frac{\sum_{i,j\in {\mathcal{C}}}a_{ij}}{\sum_{i\in {\mathcal{C}}}k_{i}}$.

\subsection{Community detection through maximum persistence probability}

The persistence probability can be used as a measure to determine the cohesiveness of a node subset, that is, to determine whether it can be interpreted as a community.
The first benchmark measure to compare it with is modularity, whose definition is
$Q_{\mathcal{C}} = \sum_{{i,j\in {\mathcal{C}}}} \left( a_{ij} - \frac{k_i k_j}{2m}\right).$
Modularity compares the presence/absence of an edge between two nodes, that is, the term $a_{ij}$, with the probability of its existence under a null hypothesis, called the configuration model. The configuration model is a random graph with the same degree sequence as the original, but whose connections are randomly rewired to destroy any form of endogenous community structure. Note that
$Q_{\mathcal{C}}$ is defined for a single cluster; however, it can be extended as a global measure. For a given partition $\Pi = \{\mathcal{C}_1,\ldots, \mathcal{C}_q \}$, the modularity of $\Pi$ is the (normalized) sum of the modularities of its clusters, and the community structure of a graph is the partition that solves the maximization problem:
\begin{equation}
\label{max_mdl}
\max_{\Pi}\, Q_{\Pi} = \max_{\Pi} \left[\frac{1}{2m}\sum_{{\mathcal{C}_{\alpha}} \in \Pi} Q_{\mathcal{C}_{\alpha}}\right].
\end{equation}
It is clear that persistence probability can play an analogous role to modularity in community detection. As in problem \eqref{max_mdl}, the network communities can be revealed by the partition that maximizes the sum of the persistences. Specifically, given a partition $\Pi$, the global persistence of $\Pi$ is the sum of the persistence probabilities of its clusters:
\begin{equation}
{\cal P}_{\Pi} = \sum_{{\mathcal{C}_{\alpha}} \in \Pi} {\cal P}_{\mathcal{C}_{\alpha}}.
\label{global persistence}
\end{equation}
Then, the community structure of a graph can be revealed by the following maximization problem:
\begin{equation} \label{maxpersistence}
\max_{\Pi}{\cal P}_{\Pi}=\max_{\Pi} \sum_{{\mathcal{C}_{\alpha}} \in \Pi} {\cal P}_{\mathcal{C}_{\alpha}}. 
\end{equation}

Problem \eqref{maxpersistence} can be formulated as a fractional integer programming problem, that can be reduced to mixed integer linear programming (the formulation is reported in the \ref{appendixA}).

The preliminary problem of finding a \textit{single} community that maximizes a slightly modified version of the persistence probability can be found in \cite{Avellone24}.
Exact and heuristic algorithms were given for that problem, and computational tests were carried out on artificial and real data. It was found that, for many test problems, the objective function studied in that paper tends to increase with respect to the size of the cluster, $|\mathcal{C}|$. This global behavior has some troublesome consequences as high-sized clusters turn out to be the best candidate solutions. To amend this bias, the authors suggested plotting the maximum of the objective function as the cluster size $|\mathcal{C}|$ varies, and selecting the local maxima rather than the global maximum, which corresponds to the trivial case $\mathcal{C} = V$.

We now replicate a similar simulation using the persistence probability. Specifically, we refer to the caveman graph shown in Figure \ref{graph_persistence}, panel (a), which represents an instance of the small-world model proposed in \cite{Watts1999}. 
The caveman graph in the figure is composed of five cliques of four nodes each, with each clique connected to two others by a single arc. The black line in the plot in panel (b) represents the persistence probability ${\cal P_C}$ as a function of the cluster size $|\mathcal{C}|$. It can be seen that ${\cal P_C}$ tends to increase. As expected, there is a first local maximum for a cluster of size $|\mathcal{C}|=4$, followed by local maxima for clusters of size multiples of $4$, which are the union of two or more cliques. Remarkably, the value of the local maxima increases with the size of the cluster. This could be problematic if we want to use the persistence probability to assess what the optimal community is. For example, for $|\mathcal{C}|=8$, ${\cal P_C}$ turns out to be larger than for $|\mathcal{C}|=4$, and this is clearly misleading, given the clique structure of the caveman graph.
Actually, one may argue that, since problem \eqref{maxpersistence} does not maximize the persistence of the single cluster but rather the sum of local persistence probabilities, local maxima of more than one clique cannot be optimal. 
The observation has a relevant consequence when one has to devise an efficient heuristic to solve problem \eqref{maxpersistence}. In fact, most heuristic algorithms for clustering find the optimal partition joining the local optima: for example, the Louvain algorithm merges nodes to clusters in a greedy way, until the cluster modularity cannot be improved. However, this strategy is precluded for the persistence probability: the local maxima of the small clusters are almost always smaller than those of the large clusters. To overcome this problem, we then advise the necessity of adjusting the persistence probability in the formulation of the problem \eqref{maxpersistence}. We call this modified objective function \textit{null-adjusted persistence}. We will formally introduce this function in the next section, but  
for illustrative purposes, in Figure \ref{graph_persistence} panel (b), we depict this new function (blue line). It can be seen that the function has a peak exactly where it ought to be, that is for the cluster of size $4$.

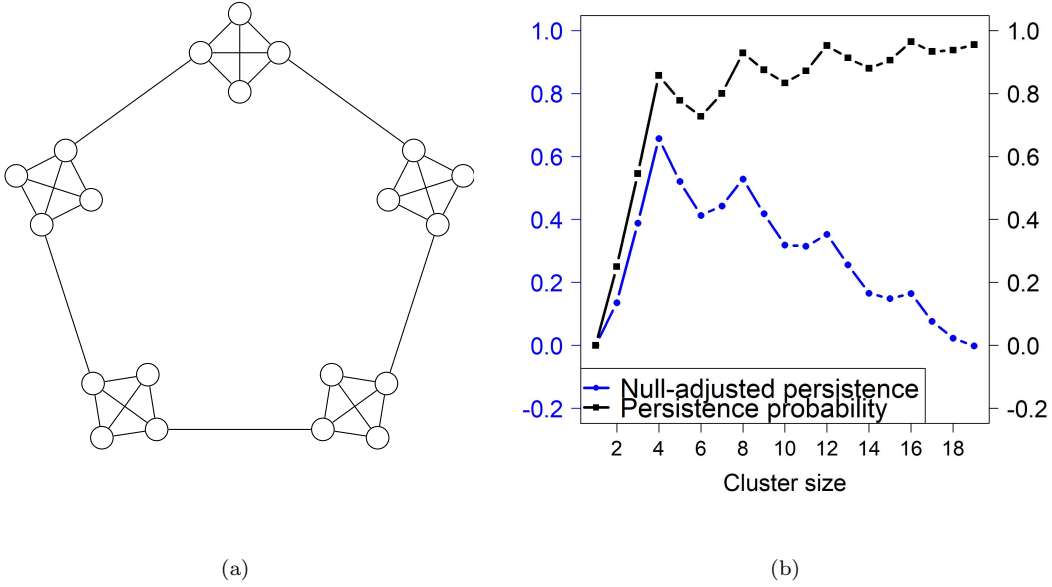
\begin{figure}[H]
	\centering
	\subfloat[]{	\begin{tikzpicture}[
			scale=0.65,
			every node/.style={draw, circle, minimum size=0.3cm, inner sep=0pt}
			]
			\node (C11) at (0.00, 4.80) {};
			\node (C12) at (0.80, 4.00) {};
			\node (C13) at (0.00, 3.20) {};
			\node (C14) at (-0.80, 4.00) {};
			
			\node (C21) at (4.57, 1.48) {};
			\node (C22) at (4.05, 0.48) {};
			\node (C23) at (3.04, 0.99) {};
			\node (C24) at (3.56, 2.00) {};
			
			\node (C31) at (2.82, -3.88) {};
			\node (C32) at (1.70, -3.71) {};
			\node (C33) at (1.88, -2.59) {};
			\node (C34) at (3.00, -2.77) {};
			
			\node (C41) at (-2.82, -3.88) {};
			\node (C42) at (-3.00, -2.77) {};
			\node (C43) at (-1.88, -2.59) {};
			\node (C44) at (-1.70, -3.71) {};
			
			\node (C51) at (-4.57, 1.48) {};
			\node (C52) at (-3.56, 2.00) {};
			\node (C53) at (-3.04, 0.99) {};
			\node (C54) at (-4.05, 0.48) {};
			
			\draw (C11) -- (C12); \draw (C11) -- (C13); \draw (C11) -- (C14);
			\draw (C12) -- (C13); \draw (C12) -- (C14); \draw (C13) -- (C14);
			
			\draw (C21) -- (C22); \draw (C21) -- (C23); \draw (C21) -- (C24);
			\draw (C22) -- (C23); \draw (C22) -- (C24); \draw (C23) -- (C24);
			
			\draw (C31) -- (C32); \draw (C31) -- (C33); \draw (C31) -- (C34);
			\draw (C32) -- (C33); \draw (C32) -- (C34); \draw (C33) -- (C34);
			
			\draw (C41) -- (C42); \draw (C41) -- (C43); \draw (C41) -- (C44);
			\draw (C42) -- (C43); \draw (C42) -- (C44); \draw (C43) -- (C44);
			
			\draw (C51) -- (C52); \draw (C51) -- (C53); \draw (C51) -- (C54);
			\draw (C52) -- (C53); \draw (C52) -- (C54); \draw (C53) -- (C54);
			
			\draw (C12) -- (C24);
			\draw (C22) -- (C34);
			\draw (C32) -- (C44);
			\draw (C42) -- (C54);
			\draw (C52) -- (C14);
			\node[rectangle, transparent] at (0, -5.5) {Caveman graph};
		\end{tikzpicture}}
	\subfloat[]{\includegraphics[width=0.50\textwidth]{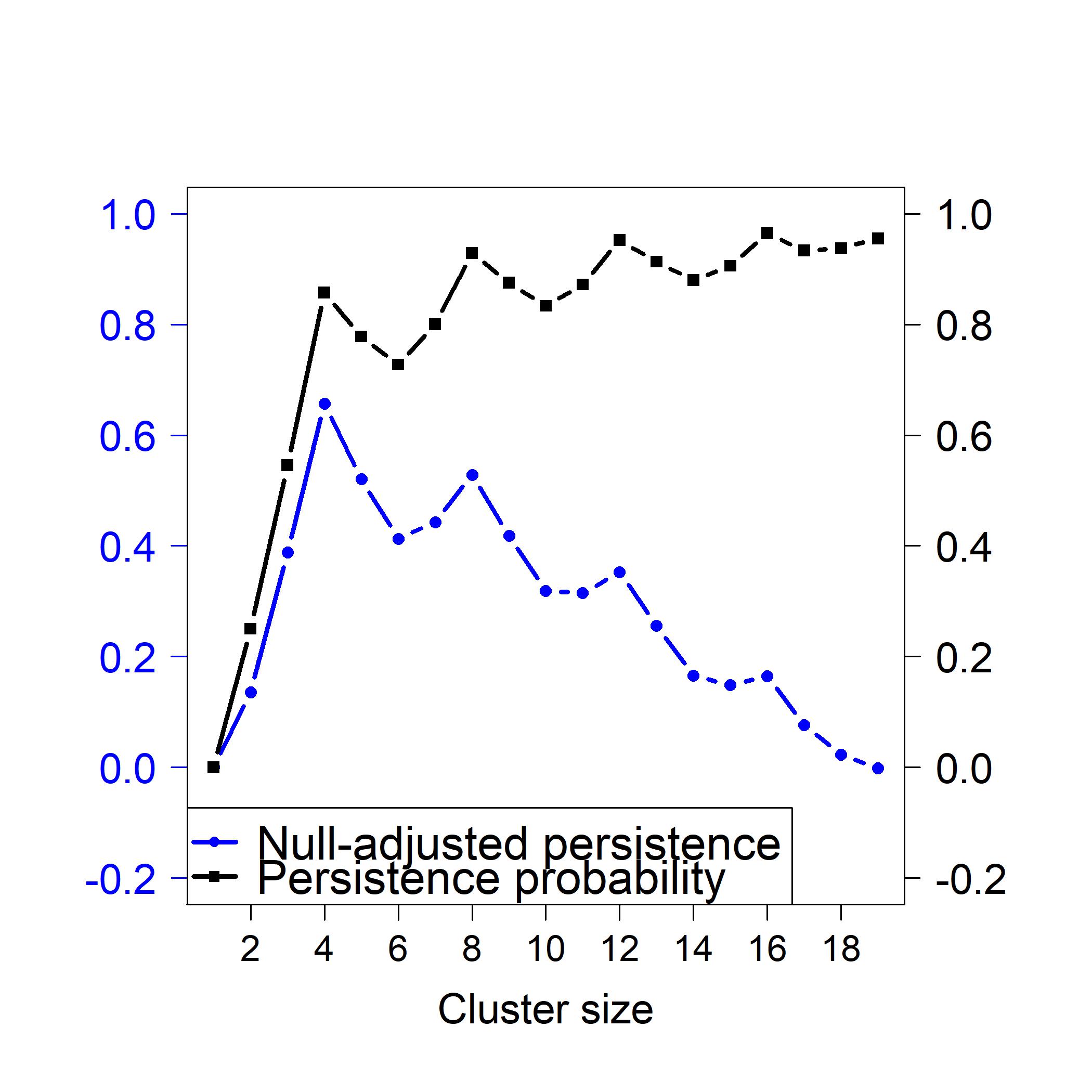}}
	\caption{Panel (a): caveman graph discussed in the text; panel (b): comparison between persistence probability (black line) and null-adjusted persistence (blue line).}
	\label{graph_persistence} 
\end{figure}

\subsection{Null-adjusted persistence function}

Drawing inspiration from the definition of modularity, null-adjusted persistence compares the persistence probability to its expected value under the configuration model. In other words, it compares the probability to the value expected under the null hypothesis, which claims that there is no community structure.

We provide the definition for the unweighted case, and from now on we assume the graph $G=(V,E)$. However, the generalization to the weighted case is straightforward.

Let $\bf B$ be the matrix of elements $b_{ij}=\frac{k_{i}k_{j}}{2m}$, where $m$ is the number of edges in the network.

\begin{definition}
The null-adjusted persistence $\persistence$ of a cluster ${\mathcal{C}}\subseteq V$ is:
\begin{equation}
	\label{definition1}
\persistence=\frac{\sum_{i\in {\mathcal{C}}}\sum_{j\in {\mathcal{C}}}a_{ij}}{\sum_{i\in {\mathcal{C}}}\sum_{j=1}^{n}a_{ij}}-\frac{\sum_{i\in {\mathcal{C}}}\sum_{j\in {\mathcal{C}}}b_{ij}}{\sum_{i\in {\mathcal{C}}}\sum_{j=1}^{n}b_{ij}}
\end{equation}
where $\bf A$ is the adjacency matrix of $G$ and $\bf B$ is the adjacency matrix of the null model. 
\end{definition}

Note that $\persistence$ is the difference between the actual persistence probability (from now on persistence, for short) of the nodes cluster ${\mathcal{C}}$ and the persistence of the same cluster under the null hypothesis of the configuration model, that is assuming a random distribution of the $m$ edges among nodes.
Let us notice that the null model that enters in formula \eqref{definition1} is the same adopted in the definition of the classical modularity introduced by \cite{2006Newman}. This null model is chosen since it ensures that the observed community structure is not simply due to variations in node degrees, and hence it provides for a natural baseline for comparing network partitions.

We can rewrite formula \eqref{definition1} in matrix form. We first introduce the following matrix
\begin{equation*}
{\bf I}_{\mathcal{C}}=
\begin{bmatrix}
	\delta_{1,{\mathcal{C}}} & 0 & 0 & \cdots & 0 \\
	0 & \delta_{2,{\mathcal{C}}} & 0 & \cdots & 0 \\
	0 & 0 & \delta_{3,{\mathcal{C}}} & \cdots & 0 \\
	\vdots & \vdots & \vdots & \ddots & \vdots \\
	0 & 0 & 0 & \cdots & \delta_{n,{\mathcal{C}}}
\end{bmatrix}
\end{equation*}
where
\begin{equation*}
\delta_{i, \mathcal{C}} =
\begin{cases}
	1 & \text{if } i \in \mathcal{C} \\
	0 & \text{if } i \notin \mathcal{C}
\end{cases}.
\end{equation*}
The matrix ${\bf I}_{\mathcal{C}}$ is the identity matrix, whose $1$'s elements on diagonal are turned $0$ for nodes that are not in $\mathcal{C}$. In particular, ${\bf I}_{\mathcal{C}}={\rm diag}\, {\bf 1}_{\mathcal{C}}$ where ${\bf 1}_{\mathcal{C}}$ is the indicator vector corresponding to the community ${\mathcal{C}}$. Therefore, we can rewrite the definition \eqref{definition1} as
\begin{equation}
\label{definition2}
\persistence =
\frac{\sum_{i,j= 1}^{n} \left({\bf I}_{\mathcal{C}} {\bf A} {\bf I}_{\mathcal{C}} \right)_{ij}}{\sum_{i,j= 1}^{n}\left( {\bf A} {\bf I}_{\mathcal{C}} \right)_{ij}}
-
\frac{\sum_{i,j= 1}^{n} \left({\bf I}_{\mathcal{C}} {\bf B} {\bf I}_{\mathcal{C}} \right)_{ij}}{\sum_{i,j= 1}^{n}\left( {\bf B} {\bf I}_{\mathcal{C}} \right)_{ij}}.
\end{equation}

We emphasize that in the first term of the Eq. \eqref{definition2} the internal arcs are counted twice, while the arcs exiting ${\mathcal{C}}$ are counted once. This is made clear by expressing the null-adjusted persistence $\persistence$ of a cluster ${\mathcal{C}} \subseteq V$ as follows:
\begin{equation}
\persistence = \frac{2m_{i}}{2m_{i}+m_{e}}-\frac{2m_{i}+m_{e}}{2m}
\label{definition3}
\end{equation}
where $m_{i}$ is the number of internal arcs in the cluster ${\mathcal{C}}$, $m_{e}$ is the number of outgoing arcs from the cluster ${\mathcal{C}}$, and $m$ is the total number of arcs in the network.

In fact, by direct computation, the first addend in Eq. \eqref{definition2} is
\begin{equation*}
	\frac{\sum_{i,j= 1}^{n} \left({\bf I}_{\mathcal{C}} {\bf A} {\bf I}_{\mathcal{C}} \right)_{ij}}{\sum_{i,j= 1}^{n}\left( {\bf A} {\bf I}_{\mathcal{C}} \right)_{ij}}=\frac{2m_{i}}{2m_{i}+m_{e}},
\end{equation*}
while the second addend in the same equation is given by
\begin{equation*}
\frac{\sum_{i,j= 1}^{n} \left({\bf I}_{\mathcal{C}} {\bf B} {\bf I}_{\mathcal{C}} \right)_{ij}}{\sum_{i,j= 1}^{n}\left( {\bf B} {\bf I}_{\mathcal{C}} \right)_{ij}}
=
\frac{\sum_{i\in {\mathcal{C}}}\sum_{j\in {\mathcal{C}}}k_{i}k_{j}}{\sum_{i\in {\mathcal{C}}}\sum_{j=1}^{n}k_{i}k_{j}}
=
\frac{\left[\sum_{i\in {\mathcal{C}}}k_{i}\right]^{2}}{\left[\sum_{i\in {\mathcal{C}}}k_{i}\right] \cdot \left[\sum_{i=1}^{n}k_{i}\right]}
=
\frac{\left(2m_{i}+m_{e} \right)^{2}}{\left(2m_{i}+m_{e} \right)\cdot 2m}
=
\frac{2m_{i}+m_{e}}{2m}.
\end{equation*}

From Eq. \eqref{definition1}, which refers to a single cluster ${\mathcal{C}}$, we can define the total null-adjusted persistence of the partition $\Pi$ as ${\cal P}^{\star}_{\Pi}=\sum_{\mathcal{C}\subseteq \Pi }\persistence$. Problem \eqref{maxpersistence} can be reformulated as follows:

\begin{equation} \label{maxNApersistence}
\max_{\Pi}{\cal P}^{\star}_{\Pi}=\max_{\Pi} \sum_{{\mathcal{C}_{\alpha}} \in \Pi} {\cal P}^*_{\mathcal{C}_{\alpha}}. 
\end{equation}

The null-adjusted persistence of the partition $\Pi$ and the persistence of the same partition are related by the following:

\begin{proposition}
\label{proposition1}
Let ${\cal P}^{\star}_{\Pi}$ be the total null-adjusted persistence of the partition $\Pi$, and ${\mathcal P}_{\Pi}$ the total persistence of the same partition, then ${\cal P}^{\star}_{\Pi}={\mathcal P}_{\Pi}-1$.
\end{proposition}

\begin{proof}
By Eq. \eqref{definition3}, we have
\begin{equation}
\begin{split}
{\cal P}^{\star}_{\Pi} & = \sum_{{\mathcal C}}\persistence
  = \sum_{{\mathcal C}} \left( \frac{2m_{i}}{2m_{i}+m_{e}}-\frac{2m_{i}+m_{e}}{2m} \right)
  = \sum_{{\mathcal C}} \frac{2m_{i}}{2m_{i}+m_{e}} - \sum_{{\mathcal C}} \frac{2m_{i}+m_{e}}{2m}\\
  & =  \sum_{{\mathcal C}} \frac{2m_{i}}{2m_{i}+m_{e}} - \frac{\sum_{{\mathcal C}}\left[\sum_{i\in {\mathcal C}}k_{i}\right]}{\sum_{i=1}^{n}k_{i}}
  ={\mathcal P}_{\Pi}-1.
\end{split}
\end{equation}
\end{proof}
This result shows that functions ${\cal P}^{\star}_{\Pi}$ and ${\mathcal P}_{\Pi}$ computed on the same partition differ by a constant, and therefore a partition that maximizes one will also maximize the other. Thus, from the point of view of finding an optimal partition, they are equivalent. Conversely, as will be shown in the next section, they exhibit very different behaviors when computed on individual clusters, which has consequences on the way in which a heuristic should be designed.

\section{Analytical results}
\label{sec3}

\subsection{Extreme values of the null-adjusted persistence}
Here, we prove some analytical results about the null-adjusted persistence. First, we focus on the extreme values of the total null-adjusted persistence ${\cal P}^{\star}_{\Pi}$.

\begin{proposition}
\label{proposition2}
Let $G$ be a network made up of $l>1$ connected components $\mathcal{C}_{\alpha}$, $\alpha=1,\dots, l$, of equal size. The total null-adjusted persistence ${\cal P}^{\star}_{\Pi}$ with respect to the natural partition into the single components $\Pi = \{\mathcal{C}_1,\ldots, \mathcal{C}_l \}$ is given by
${\cal P}^{\star}_{\Pi}=l-1$.
\end{proposition}

\begin{proof}
Let $m$ be the total number of arcs in the network. Let us refer to the partition, naturally induced by the topology of the network, into the $l$ components. Then, for each component ${\mathcal{C}}_{\alpha}$, we have $m_{i}=\frac{m}{l}$ and $m_{e}=0$.  From Eq. \eqref{definition3}, we obtain
\begin{equation}
{\mathcal P}^{\star}_{\mathcal{C}_{\alpha}}
=
\frac{2m_{i}}{2m_{i}+m_{e}}-\frac{2m_{i}+m_{e}}{2m}
=
\frac{2\frac{m}{l}}{2\frac{m}{l}}-\frac{2\frac{m}{l}}{2m}
=1-\frac{1}{l}.
\end{equation}
The total null-adjusted persistence of the partition $\Pi$ is then ${\cal P}^{\star}_{\Pi}=l\left( 1-\frac{1}{l} \right)=l-1.$
\end{proof}

Proposition \ref{proposition2} shows that ${\cal P}^{\star}_{\Pi}$ is unbounded from above, as ${\cal P}^{\star}_{\Pi}\to +\infty$ when $l\to +\infty$.
The next result provides the minimum value for ${\cal P}^{\star}_{\Pi}$. At first, recall that a $k-$partite graph is a loopless graph whose vertices can be partitioned into $k$ independent sets, that is, sets of mutually non-adjacent vertices (see \cite{Gross2013}). 

\begin{proposition}
\label{proposition3}
The total null-adjusted persistence ${\cal P}^{\star}_{\Pi}$ is bounded from below by ${\cal P}^{\star}_{\Pi} \geq -1$, and the minimum value $-1$ is attained by any multi-partite graph with respect to its canonical partition $\Pi$.
\end{proposition}
\begin{proof}
The function $\persistence = \frac{2m_{i}}{2m_{i}+m_{e}}-\frac{2m_{i}+m_{e}}{2m}$ is a strictly decreasing function of $m_{e}$. The contribution of a cluster is minimized when $m_{i}$ is zero and $m_{e}$ is as large as possible. Moreover, by Proposition \ref{proposition1},
$
{\cal P}^{\star}_{\Pi} =
\sum_{\mathcal{C}} \persistence =
\sum_{\mathcal{C}} {\mathcal{P}}_{\mathcal{C}}-1. 
$
Since ${\mathcal{P}}_{\mathcal{C}}=\frac{2m_{i}}{2m_{i}+m_{e}}$,
if $m_{i}=0$, for any ${\mathcal{C}}$, that is the graph is any multi-partite graph, then we get $\sum_{\mathcal{C}}\persistence=-1$. Moreover, if $m_{i}>0$, for some ${\mathcal{C}}$, then  $\sum_{\mathcal{C}}\persistence>-1$. This excludes the possibility that there may exist a partition with null-adjusted persistence less than $-1$.
\end{proof}
By previous results, we conclude that $-1\leq {\cal P}^{\star}_{\Pi} <+\infty$. \footnote{We recall that the minimum value of the total modularity function ${Q}_{\Pi}=\sum_{\mathcal{C}}{\rm Q}_{\mathcal{C}}$ is $-\frac{1}{2}$ and the minimum  
is attained only by a bipartite graph, when $Q_{\mathcal{C}} =-\frac{1}{4}$ for both the clusters in the natural bipartition of the network $G$.}

Let us note that higher persistence for one cluster compared to another does not guarantee greater null-adjusted persistence. For example, consider cluster ${\mathcal{C}}_{1}$ with $m_{i}=6$ and $m_{e}=4$ in a graph with $m=20$: this yields ${\mathcal{P}}_{\mathcal{C}_{1}}=0.75$ and ${\mathcal{P}}_{\mathcal{C}_{1}}^{\star}=0.35$.
For a cluster ${\mathcal{C}}_{2}$ in the same graph with $m_{i}=8$ and $m_{e}=4$, it is ${\mathcal{P}}_{\mathcal{C}_{1}}=0.80$ and ${\mathcal{P}}_{\mathcal{C}_{1}}^{\star}=0.30$. Therefore, increased persistence does not inherently correspond to higher null-adjusted persistence.

We now investigate the behavior of the null-adjusted persistence $\persistence$ in comparison with the classical modularity function computed for a cluster $\mathcal{C}$.
In the same notation as in Eq. \eqref{definition3}, the modularity $Q_{\mathcal{C}}$ of a given cluster ${\mathcal{C}}\subseteq V$ can be expressed as (see \cite{brandes2008})
\begin{equation}
Q_{\mathcal{C}} = \frac{m_{i}}{m}-\left(\frac{2m_{i}+m_{e}}{2m}\right)^2,
\label{Q}
\end{equation}
where the first term corresponds to the internal edge density and the second one to the expected edge density in the null model. Studied as functions of the single variable $m_i$ (fixing values $m$ and $m_e$) they show quite similar behaviors.
Indeed, both vanish at the same values of $m_{i}$, that is $\frac{m-m_e}{2}\pm\frac{1}{2}\sqrt{m(m-2m_e)}$, but they have a maximum at different values of $m_{i}$. 
Specifically, the local maximum of $\persistence$ is attained in $\hat{m}_{i}=\frac{1}{2}\left( \sqrt{2mm_{e}}-m_{e} \right)$ and it is equal to $\persistence(\hat{m}_{i})=1-\sqrt{\frac{2m_{e}}{m}}$. 
Conversely, the maximum for $Q_{\mathcal{C}}$ is attained in $\hat{m}_{i}=\frac{1}{2}\left( m-m_{e} \right)$ and it is equal to $Q_{\mathcal{C}}(\hat{m}_{i})=\frac{1}{4}\left( 1-\frac{2m_{e}}{m}\right)$. We conclude that while modularity $Q_{\mathcal{C}}$ is symmetric with respect to its maximum, persistence $\persistence$ is not. The behavior of the two functions $Q_{\mathcal{C}}$ and $\persistence$ with respect to the variable $m_{i}$ is depicted in Fig \ref{fig1}.
\begin{figure}[H]
\centering
\includegraphics[width=0.6\textwidth]{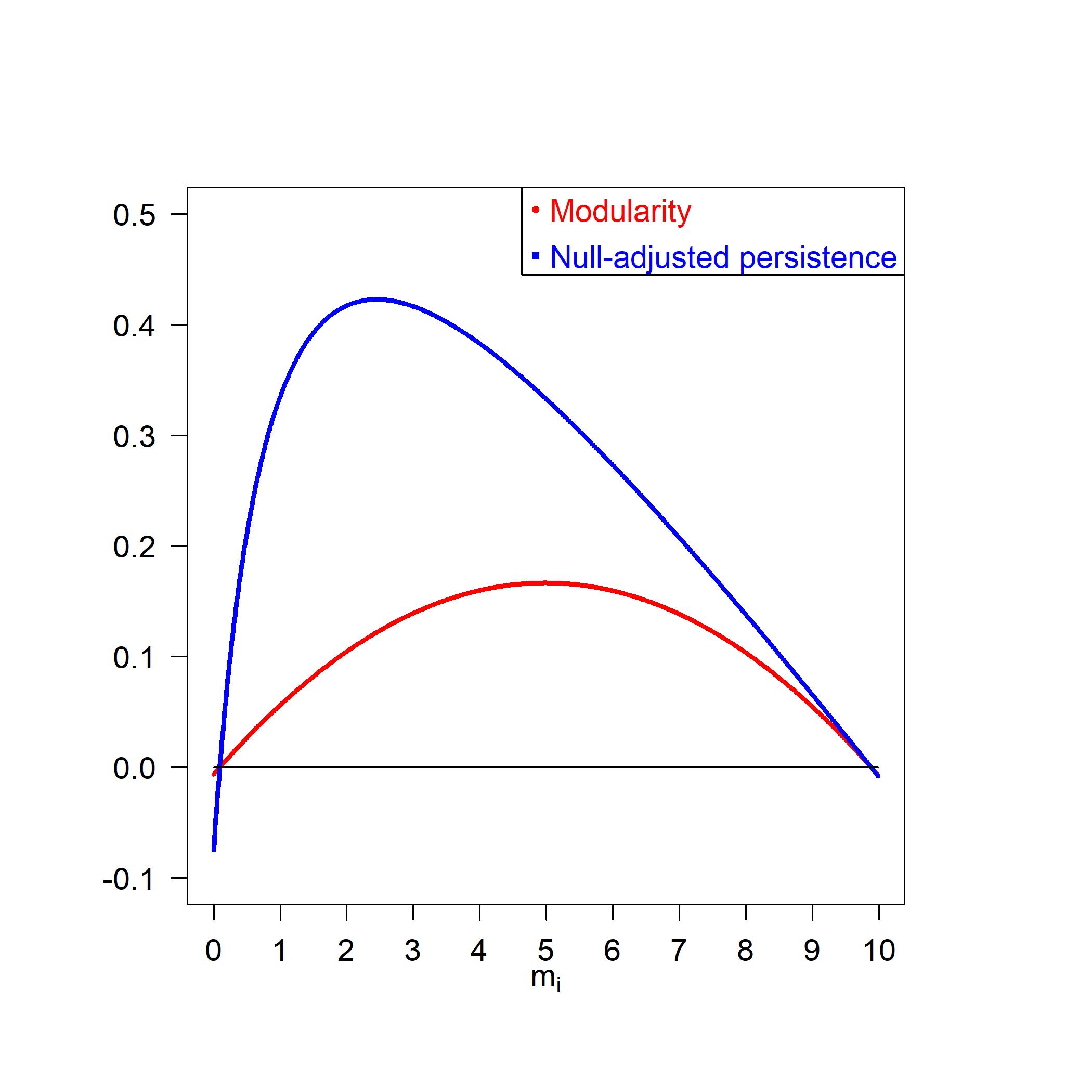}
\caption{$Q_{\mathcal{C}}$ (red line) and $\persistence$ (blue line) as functions of $m_i$. For the sake of representation, $m_{e}$ and $m$ are set to values $m_{e}=2$ and $m=12$. Both functions vanish at the same values $m_{i}=5\pm 2\sqrt{6}$. Conversely, the modularity has a maximum for $m_{i}=5$, whereas the null-adjusted persistence attains the maximum at $m_{i}=2\sqrt{3}-1=2.4641$.}
\label{fig1} 
\end{figure}
The different maxima for the two functions provide an interesting insight into our task. Indeed, if we use $Q_{\mathcal{C}}$ or $\persistence$ to find network communities, all else being equal, the maximum of the two measures is attained with a different cluster size: the null-adjusted persistence will give more weight to clusters with fewer internal arcs than modularity. Therefore, if the network contains small-sized communities, the null-adjusted persistence is more appropriate than modularity to reveal its mesoscale structure and to bring out this structure at higher resolution. Finally, we stress that, by contrast, the persistence probability $\mathcal{P}_{\mathcal{C}}$ is a monotonically increasing function with respect to $m_{i}$, and thus exhibits a behavior that, on the individual cluster, makes it not comparable to modularity.

In the next section, we highlight some further features of the null-adjusted persistence.

\subsection{Scaling behavior of the null-adjusted persistence}

In \cite{brandes2008} the authors point out that the modularity exhibits the so-called sensitivity to satellites: it identifies a clique as a natural cluster, but in presence of a clique with $l$ leafs (precisely, satellites) the optimal $Q_{\Pi}$ is attained for the clustering $\Pi$ formed by $l$ clusters. We can observe the same behavior with null-adjusted persistence, as the following example shows.
Let us consider the complete clique $K_{3}$ with leaves, represented in Fig. \ref{fig2}, panel (a). Both modularity and null-adjusted persistence disaggregate the graph into three clusters, each containing a leaf as shown in panel (b), and do not preserve the inner clique unit in panel (c).

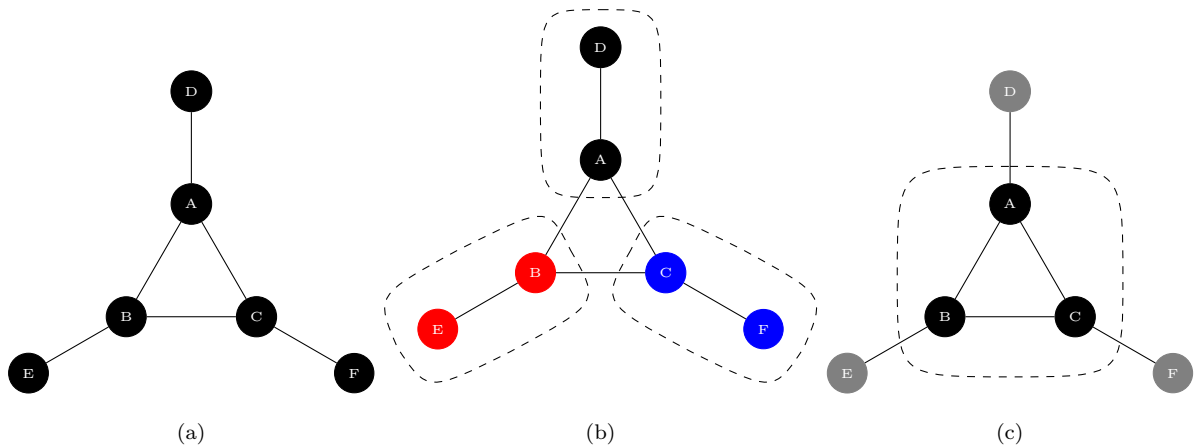
\begin{figure}[H]
\centering
\subfloat[]{\begin{tikzpicture}
\node[circle, draw=black, fill=black, minimum size=1mm, text=white] (A) at ({90+360/3 * (1 - 1)}:1cm) {\tiny A};
\node[circle, draw=black, fill=black, minimum size=1mm, text=white] (B) at ({90+360/3 * (2 - 1)}:1cm) {\tiny B};
\node[circle, draw=black, fill=black, minimum size=1mm, text=white] (C) at ({90+360/3 * (3 - 1)}:1cm) {\tiny C};
\node[circle, draw=black, fill=black, minimum size=1mm, text=white] (D) at ({90+360/3 * (1 - 1)}:2.5cm) {\tiny D};
\node[circle, draw=black, fill=black, minimum size=1mm, text=white] (E) at ({90+360/3 * (2 - 1)}:2.5cm) {\tiny E};
\node[circle, draw=black, fill=black, minimum size=1mm, text=white] (F) at ({90+360/3 * (3 - 1)}:2.5cm) {\tiny F};
\draw (A) -- (B);
\draw (A) -- (C);
\draw (B) -- (C);
\draw (A) -- (D);
\draw (B) -- (E);
\draw (C) -- (F);
\end{tikzpicture}}
\subfloat[]{\begin{tikzpicture}
\node[circle, draw=black, fill=black, minimum size=1mm, text=white] (A) at ({90+360/3 * (1 - 1)}:1cm) {\tiny A};
\node[circle, draw=red, fill=red, minimum size=1mm, text=white] (B) at ({90+360/3 * (2 - 1)}:1cm) {\tiny B};
\node[circle, draw=blue, fill=blue, minimum size=1mm, text=white] (C) at ({90+360/3 * (3 - 1)}:1cm) {\tiny C};
\node[circle, draw=black, fill=black, minimum size=1mm, text=white] (D) at ({90+360/3 * (1 - 1)}:2.5cm) {\tiny D};
\node[circle, draw=red, fill=red, minimum size=1mm, text=white] (E) at ({90+360/3 * (2 - 1)}:2.5cm) {\tiny E};
\node[circle, draw=blue, fill=blue, minimum size=1mm, text=white] (F) at ({90+360/3 * (3 - 1)}:2.5cm) {\tiny F};
\draw (A) -- (B);
\draw (A) -- (C);
\draw (B) -- (C);
\draw (A) -- (D);
\draw (B) -- (E);
\draw (C) -- (F);
\draw[dashed]
(0, 3.0) .. controls (0.8, 3.0) .. (0.8, 1.8) 
.. controls (0.8, 0.5) .. (0, 0.5) 
.. controls (-0.8, 0.5) .. (-0.8, 1.8) 
.. controls (-0.8, 3.0) .. (0, 3.0); 
\draw[dashed]
(-0.4, -0.2) .. controls (-0.0,-0.9) .. (-1.2, -1.55) 
.. controls (-2.3, -2.1) .. (-2.7, -1.5) 
.. controls (-3.0, -0.8) .. (-1.9, -0.2) 
.. controls (-0.7, 0.4) .. (-0.4, -0.2); 
\draw[dashed]
(0.4, -0.2) .. controls (0.0,-0.9) .. (1.2, -1.55) 
.. controls (2.3, -2.1) .. (2.7, -1.5) 
.. controls (3.0, -0.8) .. (1.9, -0.2) 
.. controls (0.7, 0.4) .. (0.4, -0.2); 
\end{tikzpicture}}
\subfloat[]{\begin{tikzpicture}
\node[circle, draw=black, fill=black, minimum size=1mm, text=white] (A) at ({90+360/3 * (1 - 1)}:1cm) {\tiny A};
\node[circle, draw=black, fill=black, minimum size=1mm, text=white] (B) at ({90+360/3 * (2 - 1)}:1cm) {\tiny B};
\node[circle, draw=black, fill=black, minimum size=1mm, text=white] (C) at ({90+360/3 * (3 - 1)}:1cm) {\tiny C};
\node[circle, draw=gray, fill=gray, minimum size=1mm, text=white] (D) at ({90+360/3 * (1 - 1)}:2.5cm) {\tiny D};
\node[circle, draw=gray, fill=gray, minimum size=1mm, text=white] (E) at ({90+360/3 * (2 - 1)}:2.5cm) {\tiny E};
\node[circle, draw=gray, fill=gray, minimum size=1mm, text=white] (F) at ({90+360/3 * (3 - 1)}:2.5cm) {\tiny F};
\draw (A) -- (B);
\draw (A) -- (C);
\draw (B) -- (C);
\draw (A) -- (D);
\draw (B) -- (E);
\draw (C) -- (F);
\draw[dashed]
(0, 1.5) .. controls (1.5, 1.5) and (1.5, 1.5) .. (1.5, 0) 
.. controls (1.5, -1.3) and (1.5, -1.3) .. (0, -1.3) 
.. controls (-1.5, -1.3) and (-1.5, -1.3) .. (-1.5, -0) 
.. controls (-1.5, 1.5) and (-1.5, 1.5) .. (0,1.5); 
\end{tikzpicture}}
\caption{Connected network $G$ formed by the inner clique $K_{3}$ and three leaves (panels (a) and (c)). Optimal partition of the network $G$ according to modularity and null-adjusted persistence (panel (b)).}
\label{fig2} 
\end{figure}

\begin{figure}[H]
	\centering
	\includegraphics[width=0.6\textwidth]{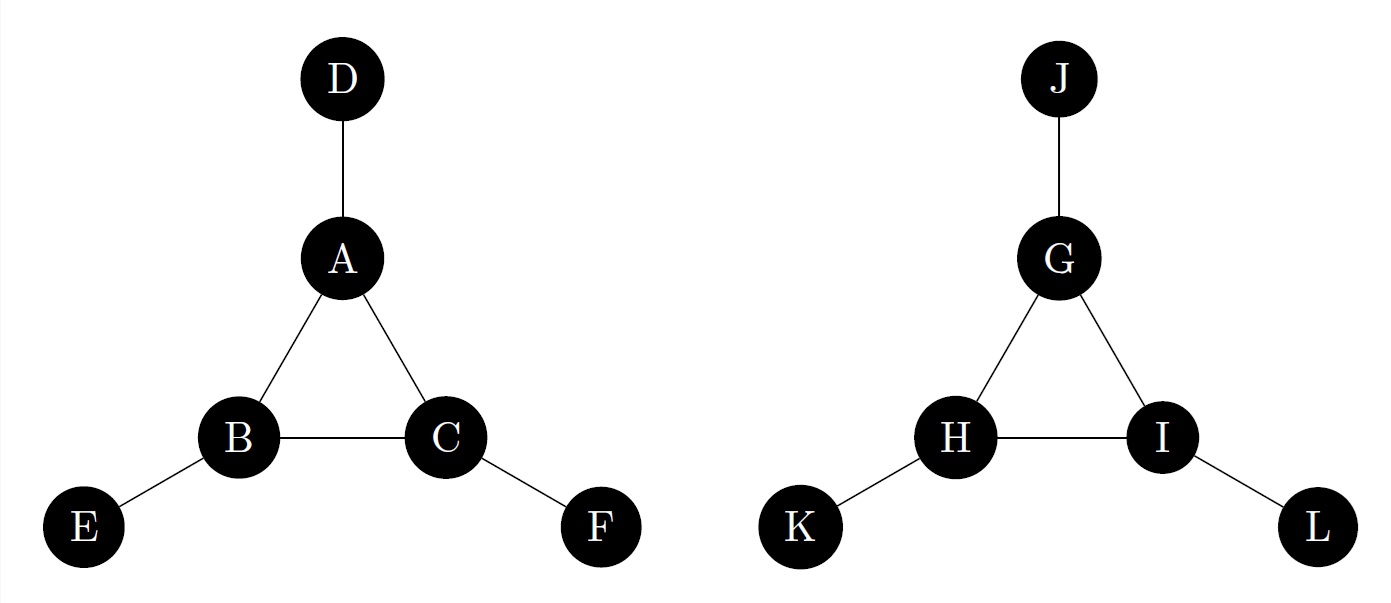}
	\caption{Disconnected network $G$ formed by two cliques $K_{3}$ with leaves.}
	\label{fig3} 
\end{figure}
However, if we consider a non-connected network $G$ originally composed by two identical and disjoint cliques $K_{3}$ with leaves, (see Fig \ref{fig3}), the modularity identifies an optimal partition that keeps each clique in a single cluster containing the satellites, too. This result
conflicts with what previously found, in which each clique was split into three components (see Fig. \ref{fig4}). In other words, as also \cite{brandes2008} point out, modularity does not exhibit a scale-invariant behavior.
\begin{figure}[H]
	\centering
	\includegraphics[width=0.6\textwidth]{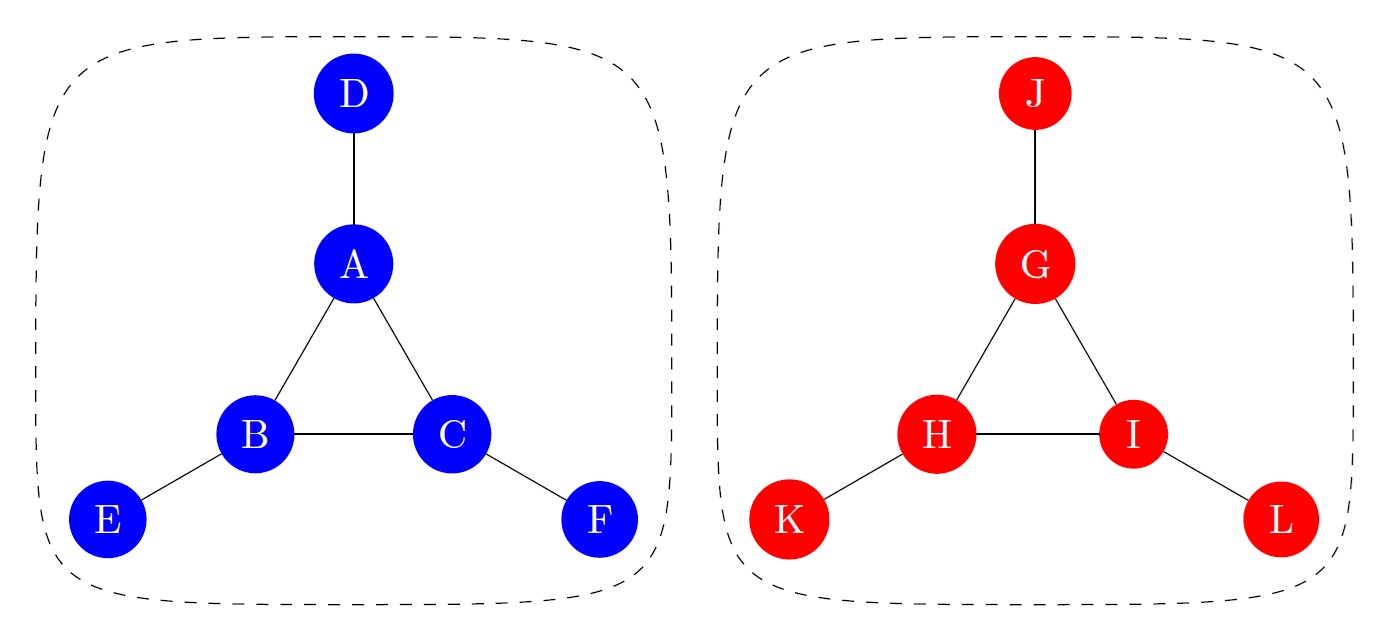}
	\caption{Optimal partition $\Pi$ of the disconnected network $G$ formed by two cliques $K_{3}$ with leaves, according to modularity $Q(\Pi)$).}
	\label{fig4} 
\end{figure}
Conversely, for the null-adjusted persistence, the optimal partition in which each clique is divided into three components is preserved even if we double the clique, as shown in Fig. \ref{fig5}, suggesting that  the null-adjusted persistence {\it is} scale invariant.
\begin{figure}[H]
	\centering
	\includegraphics[width=0.6\textwidth]{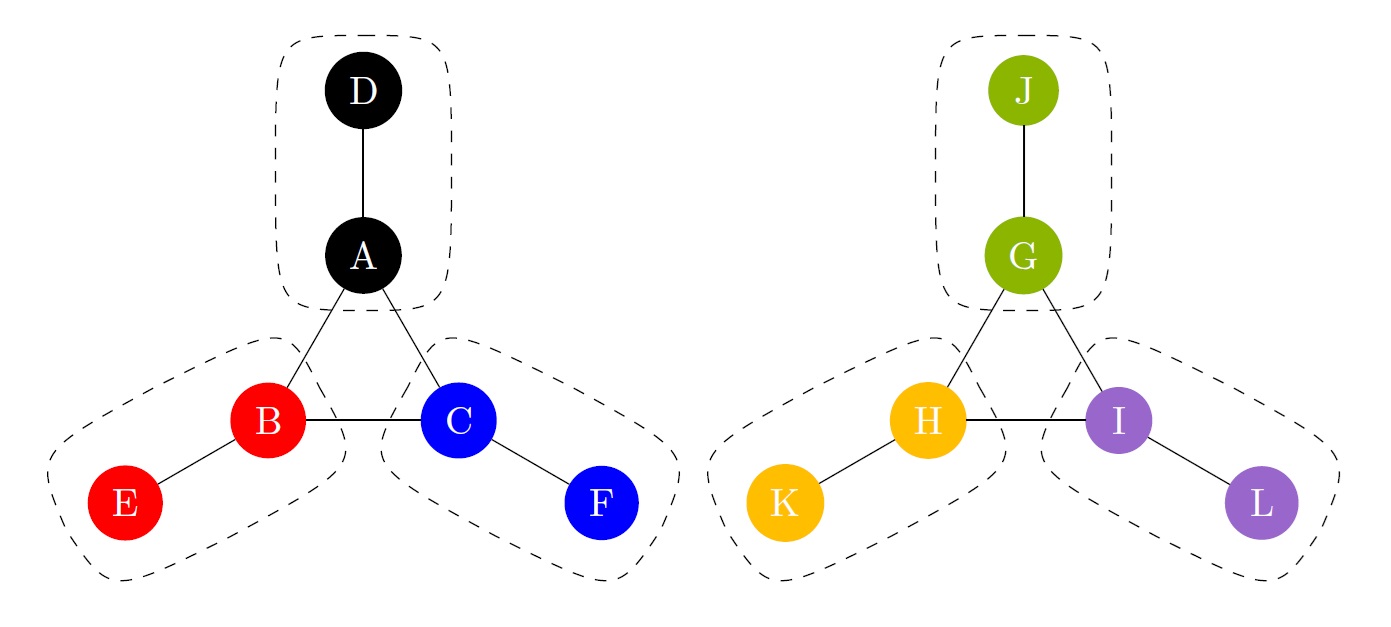}
	\caption{Optimal partition of the two cliques $K_{3}$ with leaves according to the null-adjusted persistence ${\cal P}^{\star}_{\Pi}$.}
	\label{fig5} 
\end{figure}

\subsection{The Resolution Limit}
One of the problems often pointed out with modularity is that it has an intrinsic scale dependence, which limits the number and size of modules it can detect. We show in the next result that the null-adjusted persistence does not suffer from this limitation.

\begin{proposition}
\label{resolution_prop}
Let $G=(V,E)$ be a connected graph of $l>1$, $l$ even, identical cliques ${\mathcal{C}}$ connected in a circle, in such a way that each pair of adjacent cliques is connected by a single edge. Let ${\Pi}_{1}$ be the partition of $G$ into the $l$ cliques and ${\Pi}_{2}$ be the partition of $G$ into the $\frac{l}{2}$ pairs of adjacent cliques. Then ${\cal P}^{\star}_{\Pi_1}>{\cal P}^{\star}_{\Pi_2}$.
\end{proposition}
\begin{proof}
Let $k$ be the number of arcs inside each clique ${\mathcal{C}}$ and consider the partition ${\Pi}_{1}$.
In this case, $m_{i}=k$, $m_{e}=2$ and $m=l(k+1)$.
By formula \eqref{definition3} we obtain,  for each cluster ${\mathcal{C}}$,
$$\persistence
=\frac{2k}{2k+2}-\frac{2k+2}{2l(k+1)}=\frac{k}{k+1}-\frac{1}{l}.$$
The total null-adjusted persistence of the partition ${\Pi}_{1}$ is then
${\cal P}^{\star}_{\Pi_1}=\frac{kl}{k+1}-1$.

Consider now the partition ${\Pi}_{2}$: $m_{i}=2k+1$, $m_{e}=2$ and $m=l(k+1)$.
By Formula \eqref{definition3}, we obtain, for each cluster ${\mathcal{C}}$,
$$\persistence
=
\frac{2(2k+1)}{2(2k+1)+2}-\frac{2(2k+1)+2}{2l(k+1)}
=
\frac{2k+1}{2k+2}-\frac{2}{l}.
$$
The total null-adjusted persistence of the partition ${\Pi}_{2}$ is then
${\cal P}^{\star}_{\Pi_2}
=
\frac{l}{2} \left[ \frac{2k+1}{2k+2}-\frac{2}{l} \right]
=
\frac{(2k+1)l}{4(k+1)}-1
$
Since the inequality
\begin{equation*}
	\frac{kl}{k+1}-1>\frac{(2k+1)l}{4(k+1)}-1
\end{equation*}
equals $k>\frac{1}{2}$, which is satisfied for any $l$, we can conclude that ${\cal P}^{\star}_{\Pi_1} > {\cal P}^{\star}_{\Pi_2}$.
\end{proof}
Clusters belonging to partitions ${\Pi}_{1}$ and ${\Pi}_{2}$ of Proposition \ref{resolution_prop} are represented in Fig. \ref{fig7}.
\begin{figure}[H]
	\centering
        \includegraphics[width=0.41\textwidth]{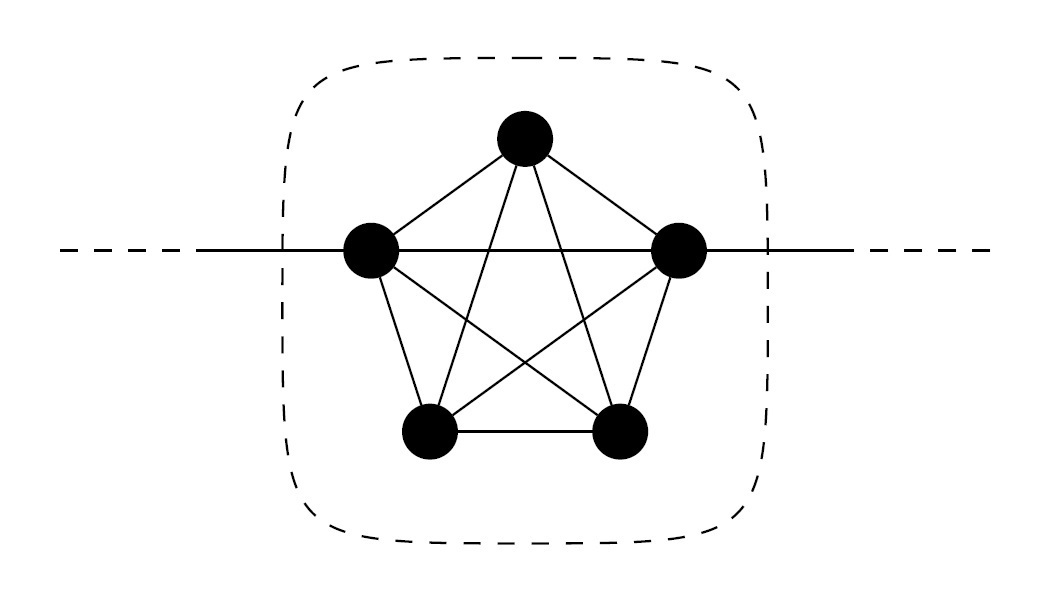}
        \includegraphics[width=0.54\textwidth]{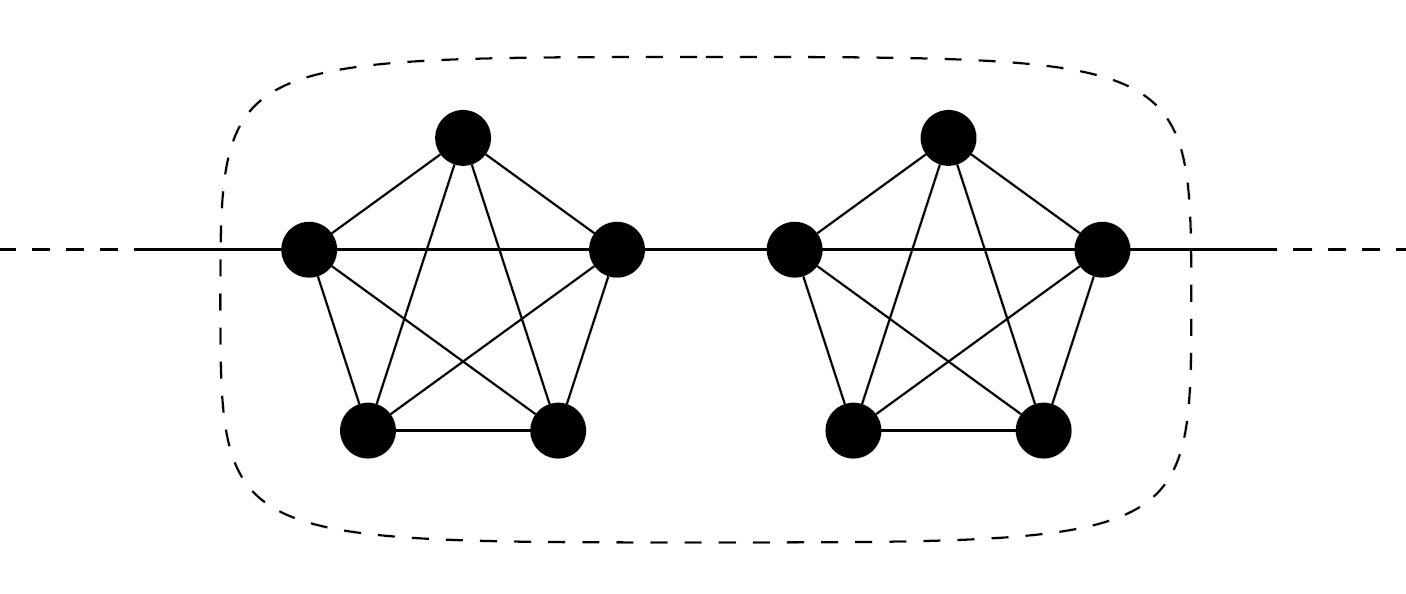}
 	\caption{Clusters in partition ${\Pi}_{1}$ (left) and in ${\Pi}_{2}$ (right).}
	\label{fig7} 
\end{figure}

A similar inequality for modularity function has been obtained in \cite{brandes2008}. In particular, under the same hypothesis -- i.e. considering the same partitions ${\Pi}_{1}$ and ${\Pi}_{2}$ -- the authors show that
${Q}_{\Pi_1}> {Q}_{\Pi_2}$ if $l<\sqrt{2m}$. 
We provide here an alternative proof. 
In the clustering ${\Pi}_{1}$, for each cluster $\mathcal{C}$,
$
{Q}_{\mathcal{C}}
=
\frac{k}{l(k+1)}-\frac{1}{l^2}
$
so that the total modularity is
$
{Q}_{{\Pi}_{1}}
=
l \left[ \frac{k}{l(k+1)}-\frac{1}{l^2} \right]
=
\frac{k}{k+1}-\frac{1}{l}.
$
In the second clustering ${\Pi}_{2}$, for each cluster $\mathcal{C}$, 
$
{Q}_{\mathcal{C}}
=
\frac{2k+1}{l(k+1)}-\frac{4}{l^2}
$
so that the total modularity is
$
{Q}_{{\Pi}_{2}}
=
\frac{l}{2} \left[ 	\frac{2k+1}{l(k+1)}-\frac{4}{l^2} \right]
=
\frac{2k+1}{2k+2}-\frac{2}{l}.
$
Therefore, the clustering ${\Pi}_{1}$ has higher modularity than the clustering ${\Pi}_{2}$, that is $Q_{{\Pi}_{1}}>Q_{{\Pi}_{2}}$, if
\begin{equation*}
\frac{k}{k+1}-\frac{1}{l}>\frac{2k+1}{2k+2}-\frac{2}{l}
\end{equation*}
solved for $l<2(k+1)$. Since $m=l(k+1)$, this condition equals $l<\frac{2m}{l}$, equivalent to $l<\sqrt{2m}$.

It is worth noting that this result is dependent on the number of cliques $l$ and their size $k$. The modularity is able to discriminate cliques containing $k$ arcs only if the number $l$ of cliques does not exceed $2(k+1)$, and it fails when the number of cliques becomes large enough compared to the number of arcs contained in each clique. 
By contrast, the total null-adjusted persistence is able to discriminate cliques in the network regardless of their number, and, in the end, regardless of the size $m$ of the network itself, overcoming the resolution limit typical of the modularity function.

\subsection{Merging clusters}
A critical consideration in community detection is evaluating the benefit of merging clusters. Merging is beneficial if it improves the objective function. In this section, we quantify the gain or cost, in terms of null-adjusted persistence, in merging two distinct clusters.
Let ${\mathcal{C}}_{1}$ and ${\mathcal{C}}_{2}$ be two clusters with internal arcs ${m_{i}^{(1)}}$ and ${m_{i}^{(2)}}$, external arcs ${m_{e}^{(1)}}$ and ${m_{e}^{(2)}}$, and ${m_{e}^{(12)}}$ arcs in between connecting them. The difference between the null-adjusted persistence of the merged cluster ${\mathcal{C}}$ and the sum of those of the two separate clusters ${\mathcal{C}}_{1}$ and ${\mathcal{C}}_{2}$,
$\Delta \persistence=\persistence-\left( \persistence_{1}+\persistence_{2} \right)$, quantifies the merging gain or cost, and allows us to provide a threshold above which the merging operation is convenient, as we show in the following:
\begin{proposition}
	\label{proposition5}
The null--adjusted persistence is increased by merging two clusters, i.e. $\Delta \persistence>0$, if and only if 
\begin{equation}
{m_{e}^{(12)}}>
\frac{2{m_{i}^{(2)}}+{m_{e}^{(2)}}}{2{m_{i}^{(1)}}+{m_{e}^{(1)}}}\, {m_{i}^{(1)}}+
\frac{2{m_{i}^{(1)}}+{m_{e}^{(1)}}}{2{m_{i}^{(2)}}+{m_{e}^{(2)}}}\, {m_{i}^{(2)}}.
\label{thresholdP}
\end{equation}
\end{proposition}
\begin{proof}
Let $m_{i}$ and $m_{e}$ be the internal and external arcs of the merged cluster ${\mathcal{C}}$, respectively. We have: $m_{i}={m_{i}^{(1)}}+{m_{i}^{(2)}}+{m_{e}^{(12)}}$ and $m_{e}={m_{e}^{(1)}}+{m_{e}^{(2)}}-2{m_{e}^{(12)}}$.
Therefore
\begin{equation*}
\begin{split}
\Delta \persistence &=\persistence-\left( \persistence_{1}+\persistence_{2} \right)\\
&=\frac{2\left( {m_{i}^{(1)}}+{m_{i}^{(2)}}+{m_{e}^{(12)}}\right)}
{\left( 2{m_{i}^{(1)}}+{m_{e}^{(1)}}\right)+\left( 2{m_{i}^{(2)}}+{m_{e}^{(2)}}\right)}-
\frac{2{m_{i}^{(1)}}}{2{m_{i}^{(1)}}+{m_{e}^{(1)}}}-
\frac{2{m_{i}^{(2)}}}{2{m_{i}^{(2)}}+{m_{e}^{(2)}}}
\end{split}
\end{equation*}
By solving $\Delta \persistence>0$ with respect to ${m_{e}^{(12)}}$ the inequality \eqref{thresholdP} immediately follows.
\end{proof}
Proposition \ref{proposition5} implicitly states that the threshold on the value of ${m_{e}^{(12)}}$ beyond which it becomes convenient to merge the two clusters does not depend on the size $m$ of the overall network. This further supports the scale invariance of the null-adjusted persistence. Notice that the same does not hold for modularity. In fact, if we {compute the merging cost for modularity, that is we solve $\Delta Q_{\mathcal{C}}=Q_{\mathcal{C}}-\left( Q_{{\mathcal{C}}_{1}}+Q_{{\mathcal{C}}_{2}} \right)>0$, we get
\begin{equation}
	{m_{e}^{(12)}}>
	\frac{ \left( 2{m_{i}^{(1)}}+{m_{e}^{(1)}} \right) \left( 2{m_{i}^{(2)}}+{m_{e}^{(2)}}\right) }{2m}
	\label{thresholdQ}
\end{equation}
which depends on the size $m$ of the network, confirming the scale dependence behavior. 

Fig. \ref{fig8} illustrates the result of Proposition \ref{proposition5}. In this case ${m_{i}^{(1)}}=4$, ${m_{i}^{(2)}}=3$, ${m_{e}^{(1)}}=7$, ${m_{e}^{(2)}}=6$ and ${m_{e}^{(12)}}=3$, and we would need at least ${m_{e}^{(12)}}=\lceil \frac{139}{20} \rceil = \lceil 6.95 \rceil =7$ arcs between  ${\mathcal{C}}_{1}$ and  ${\mathcal{C}}_{2}$ to conveniently merge them into a single cluster ${\mathcal{C}}$. Note that we cannot measure the merging cost for modularity without knowing the number of arcs in the network, $m$.
\begin{figure}[H]
	\centering
	\includegraphics[width=0.6\textwidth]{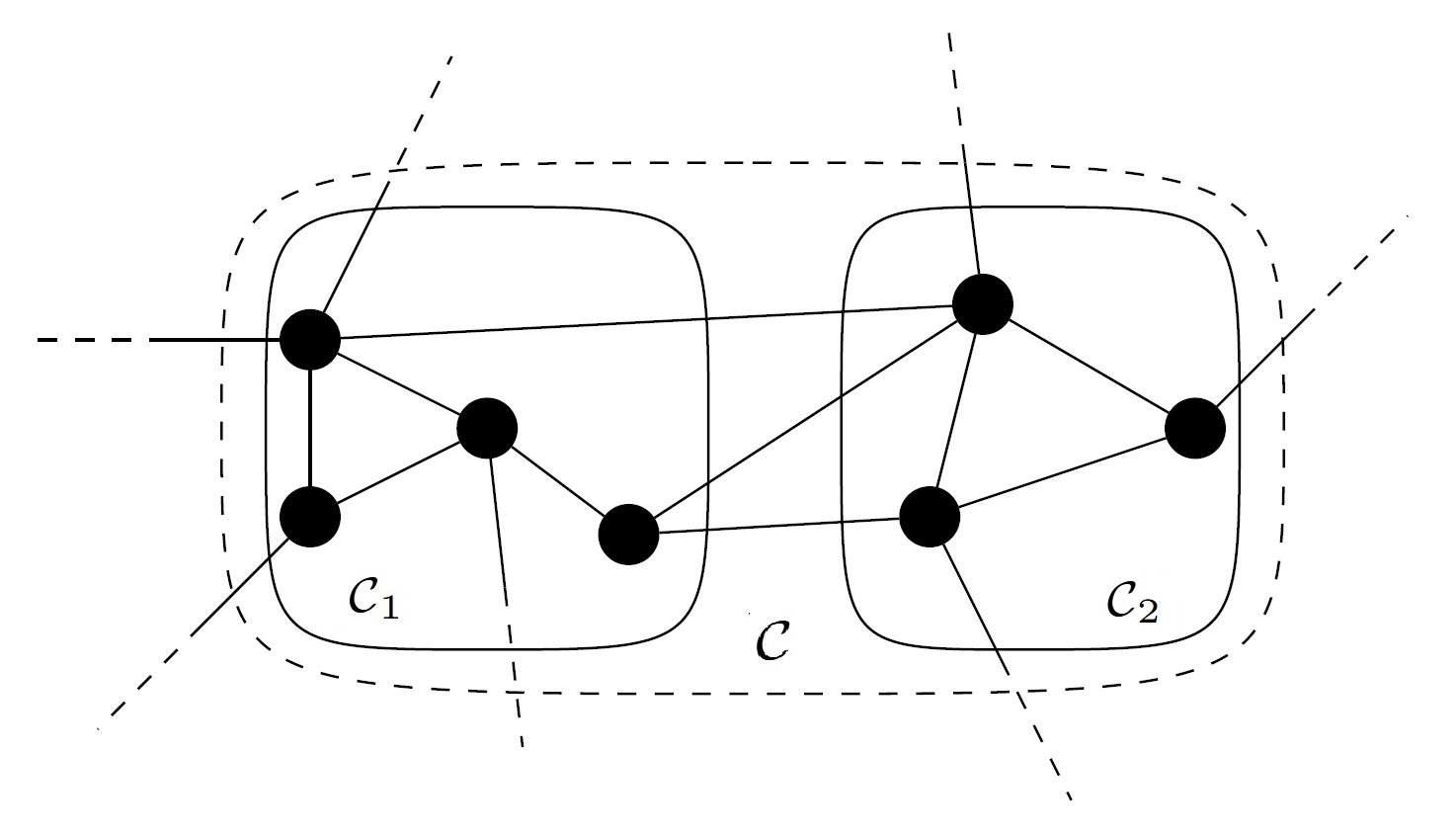}
	\caption{Illustrative example of the merging condition of two clusters.}
	\label{fig8} 
\end{figure}

Proposition \ref{proposition5} further justifies why null-adjusted persistence does not suffer from the same resolution limit as modularity. Indeed, under the hypotheses of Proposition \ref{resolution_prop}, $m_{i}=k$, $m_{e}=2$ for both cliques and $m=l(k+1)$, so that in the case of the null-adjusted persistence, ${m_{e}^{(12)}}>2k$. This shows how difficult it is to merge two cliques when using the null-adjusted persistence and how this becomes increasingly difficult as the size of the two cliques increases. On the contrary, using modularity, ${m_{e}^{(12)}}>\frac{2(k+1)}{l}$. The threshold \eqref{thresholdQ} for modularity becomes less than $1$ when $l>2(k+1)$ so that the presence of a single link between the two clusters is enough to make it convenient to merge them.

\section{Testing null-adjusted persistence optimization on benchmark and real networks}
\label{sec4}
In this section, we test the null-adjusted persistence as a quality function to detect communities first on two classes of simulated graphs and then on a real-world network. The scope is to assess its ability to catch the mesoscale structure, in particular, on classes of synthetic networks that provide controlled environments where the underlying community structure is known, and on ground-truth data, allowing rigorous testing of how well algorithms can identify and distinguish these communities.

As can be seen in the appendix, solving for maximum persistence through integer fractional programming is only viable for small graphs. Therefore, a heuristic algorithm must be devised for our tests. For comparison purposes, we adapted the classical Louvain algorithm for modularity (\cite{2008Blondel}) to the new objective function. Actually, the principle by which two clusters merge is the same for both methods; the only difference is the use of null-adjusted persistence instead of modularity. The detailed description of the proposed algorithm is reported in the \ref{appendixB}. The algorithm has been implemented in \textsc{C++} and the \proglang{R} package \pkg{persistence} has been developed for the computation of the null-adjusted persistence.

\subsection {Community detection on benchmark networks}

We begin by testing the null-adjusted persistence function on two different classes of benchmark networks. The first one is the class of the caveman graphs (see \cite{Watts1999}), progressively modified through random rewiring. These graphs are built from cliques -- fully connected subgraphs -- that represent tightly-knit communities or \textit{caves}, by shifting one edge from each clique and using it to connect to a neighboring one. These cliques are loosely connected through sparse inter-community arcs and, hence, exhibit clear and easily identifiable community boundaries, that are progressively hidden by rewiring. Caveman graphs are particularly useful for testing algorithms under idealized conditions where communities are densely connected internally but sparsely linked externally.
The second one is the class of simulated networks generated according to the methodology proposed by \cite{Lancichinetti2008}. This procedure generates networks that are as close as possible to real networks, which are often characterized
by a highly variable node degree. The Lancichinetti–Fortunato–Radicchi (LFR) networks provide a more complex and realistic scenario than caveman graphs. Designed to mimic the structure of real-world networks, LFR graphs feature power-law degree distributions and communities of varying sizes. A central element of these graphs is the mixing parameter $\mu$, which determines how many of a node's arcs connect to nodes outside its community.

For both classes of networks, we compare the performance of the null-adjusted persistence against the classical modularity function. We consider a range of conditions and evaluate the performance of the two methods using two well-known measures of partition similarity: the Adjusted Rand Index (ARI) and the Normalized Mutual Information (NMI) (\cite{Hubert1985,Danon2005}).
Both indices range from 0 to 1, where 0 indicates a random assignment of nodes to community, whereas 1 indicates a perfect match between the partitions.

Figure \ref{fig9} depicts results for caveman graphs and compares the indices in dependence on the edge rewiring. Specifically, the initial configuration -- the classical caveman structure -- is progressively modified according to a mechanism of rewiring consisting of moving a randomly selected link while preserving the degree distribution.  Panels (a) and (b) refer to a network of $60$ nodes, distributed into 12 communities, each of 5 nodes. Initially, both objective functions can intercept the community structure underlying the graph. Null-adjusted persistence maintains this capability well even after the rewiring mechanism has produced substantial link shuffling. Panels (c) and (d) refer to a network of $120$ nodes distributed into 24 communities, each of 5 nodes. In this case, the modularity function fails in identifying the initial community structure, as the number of communities overcomes the threshold 
\eqref{thresholdQ}, and it becomes convenient to merge adjacent clusters.
Conversely, the null-adjusted persistence is still able to identify the community structure, as was theoretically predicted in Proposition \ref{resolution_prop}.

\begin{figure}[H]
	\centering
	\subfloat[]{\includegraphics[width=0.45\textwidth]{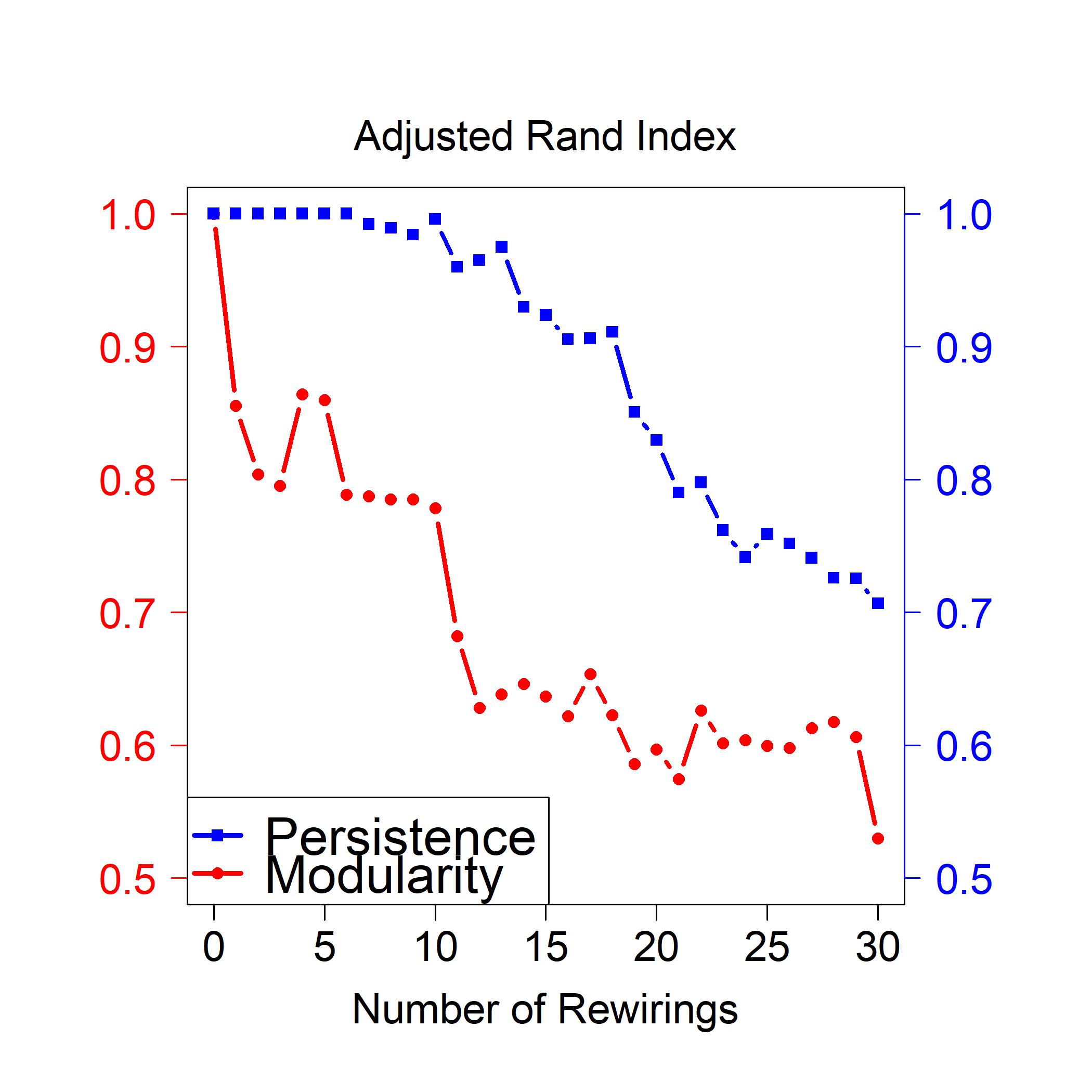}}
	\subfloat[]{\includegraphics[width=0.45\textwidth]{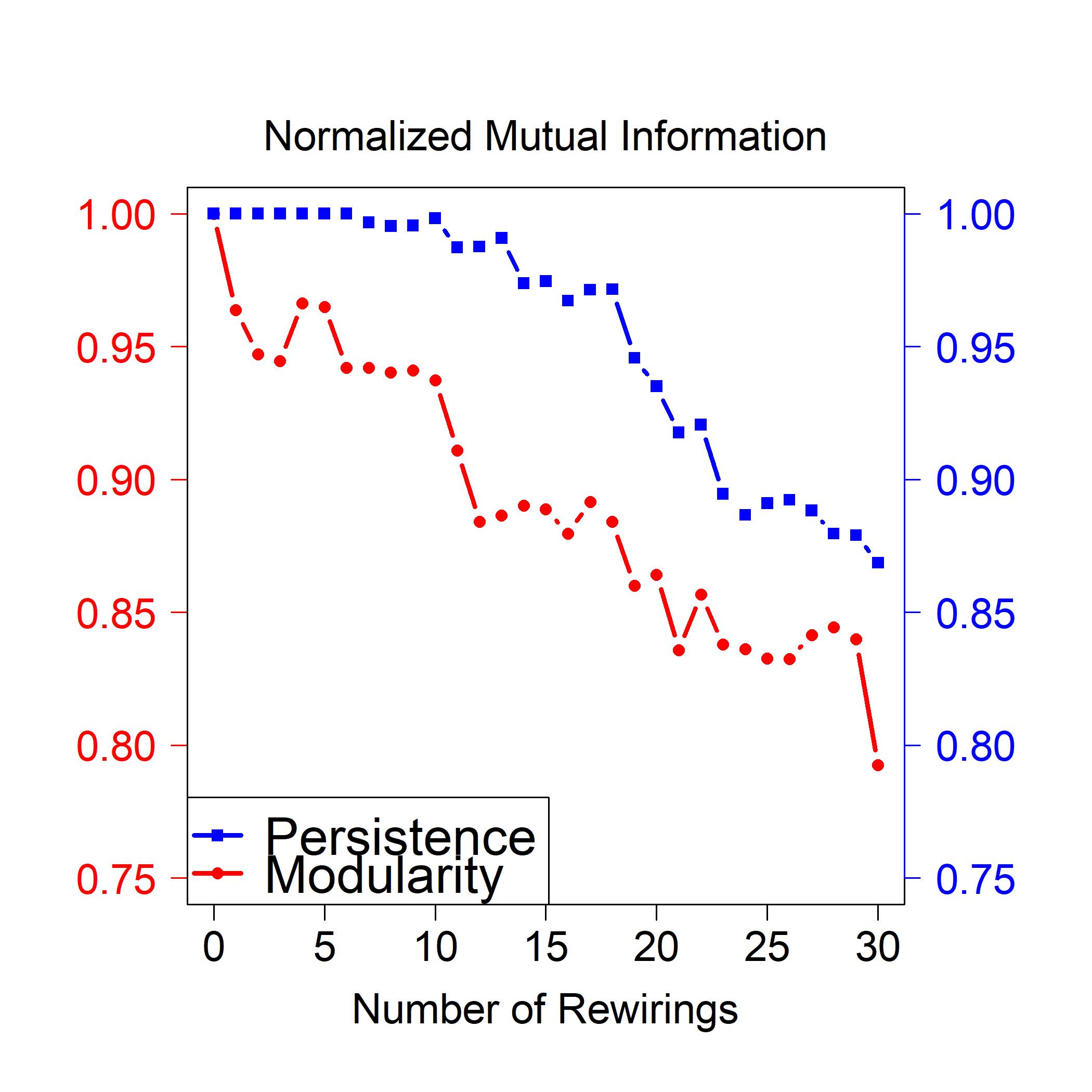}}\\
	\subfloat[]{\includegraphics[width=0.45\textwidth]{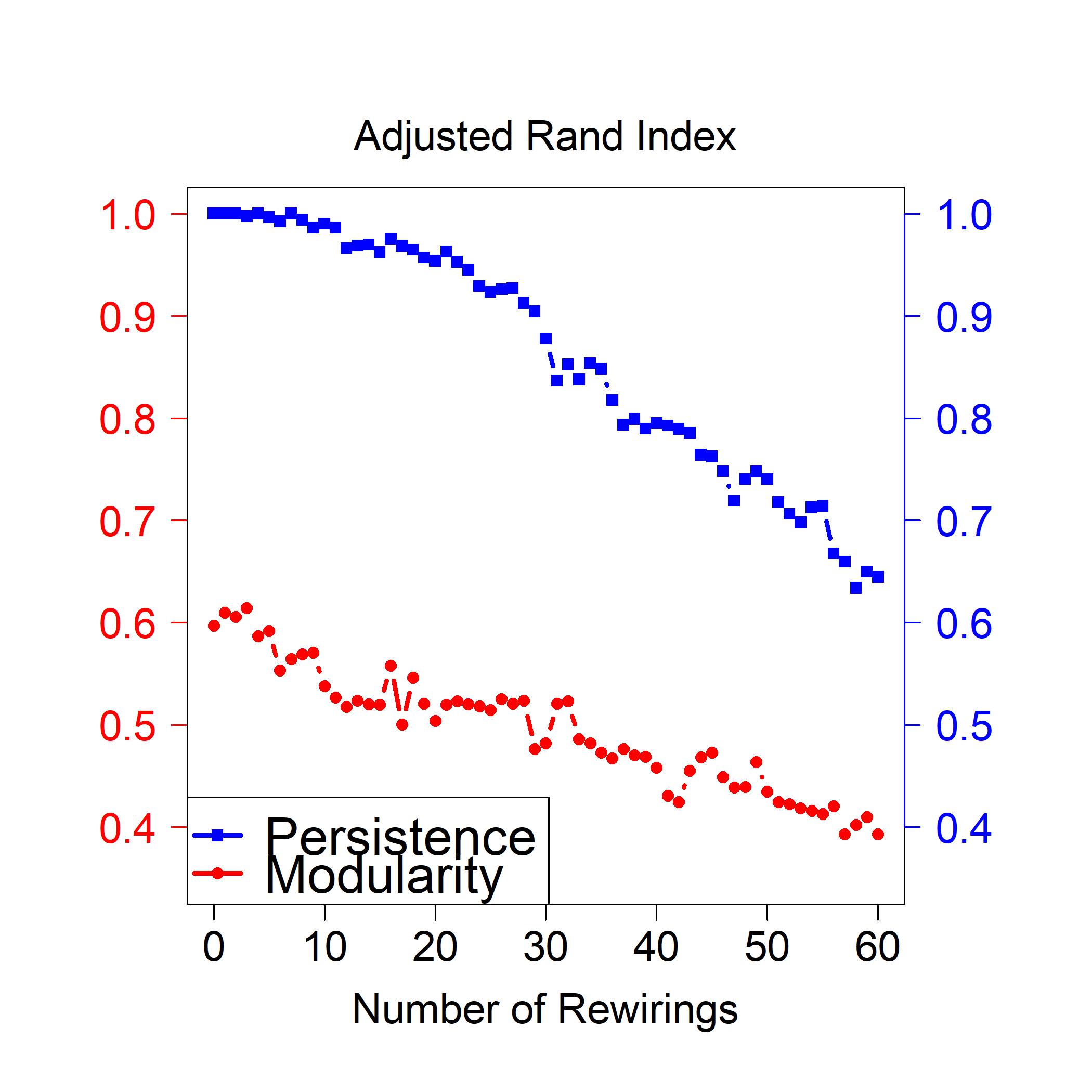}}
	\subfloat[]{\includegraphics[width=0.45\textwidth]{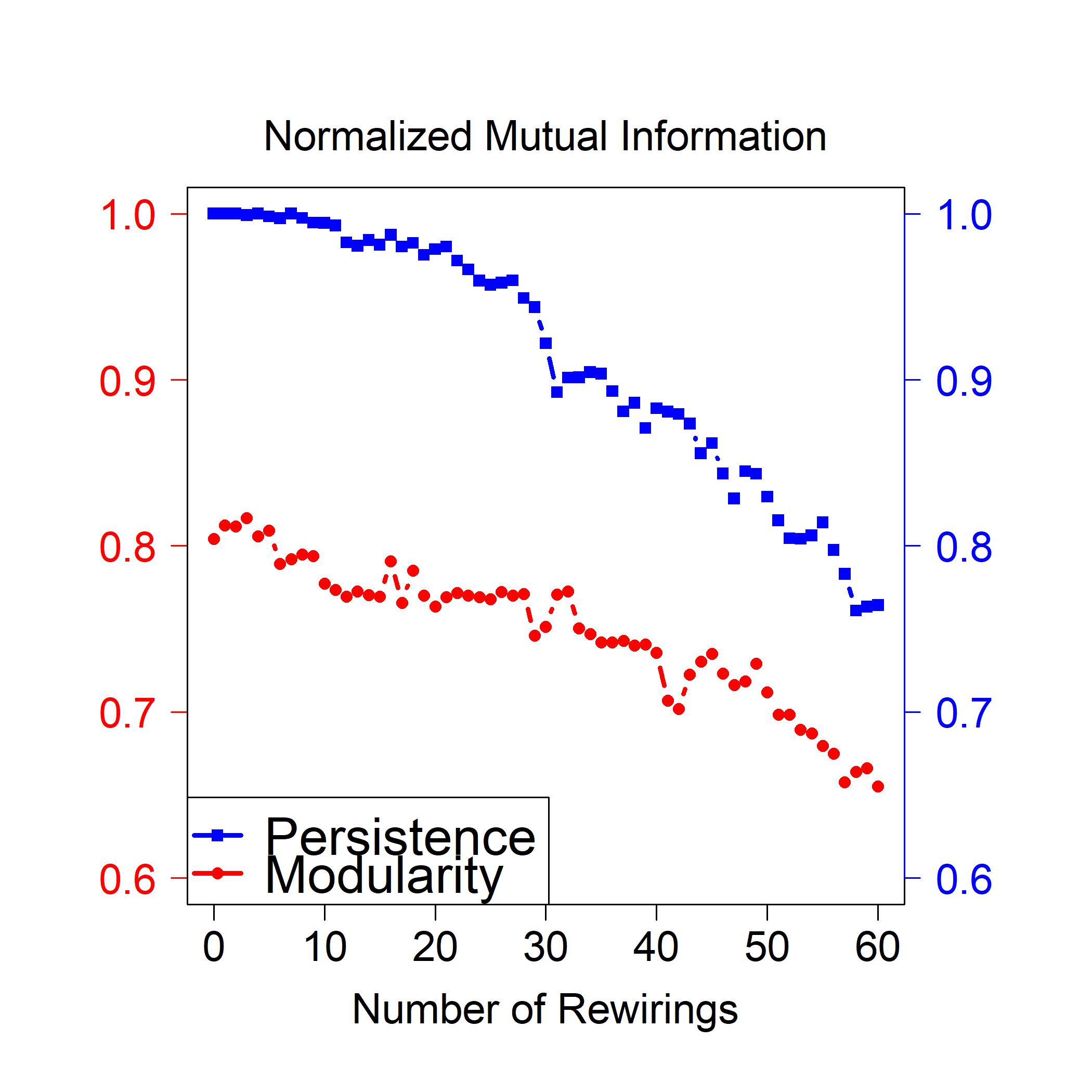}}\\
	\caption{Adjusted Rand Index and Normalized Mutual Information for caveman graphs with 12 cliques of 5 nodes (panels (a) and (b)) and 24 cliques of 5 nodes (panels (c)and (d))}
	\label{fig9} 
\end{figure}

Increasing the network size, we conducted a similar study 
on LFR graphs with $n=1000$ nodes.  In Table \ref{Table1} we list the results for LFR graphs whose degree and community size power law distributions have exponents $\tau_{1}=2$ and $\tau_{2}=2$, respectively. 
As can be seen, both the ARI and the NMR values for $\mathcal{P}^{\star}$ are slightly higher than for $\mathcal{Q}$,
confirming the ability of the null-adjusted persistence to capture the community structure, for both average degrees $10$ and $15$. In all cases, a significant independence of the true partition detection capability from the value of $\mu$ emerges.

\begin{table}[h]
	\begin{center}
\begin{tabular}{c|cccc|cccc|}
\cline{2-9}
 & \multicolumn{4}{c|}{ARI} & \multicolumn{4}{c|}{NMR} \\
\cline{2-9}
 & \multicolumn{2}{c|}{Av Deg 10} & \multicolumn{2}{c|}{Av Deg 15}
 & \multicolumn{2}{c|}{Av Deg 10} & \multicolumn{2}{c|}{Av Deg 15} \\
\hline
$\mu$ & $\mathcal{Q}$ & ${\mathcal{P}}^{\star}$ & $\mathcal{Q}$  & ${\mathcal{P}}^{\star}$ & $\mathcal{Q}$ & ${\mathcal{P}}^{\star}$ & $\mathcal{Q}$ & ${\mathcal{P}}^{\star}$ \\
\hline
$0.1$ & $0.985888$ & $0.998525$ & $0.998935$ & $0.999773$ &
$0.995314$ & $0.999243$  & $0.999610$ & $0.999913$ \\
\hline
$0.2$ & $0.984673$ & $0.997420$ & $0.998819$ & $0.999634$ &
$0.994874$ & $0.998991$ & $0.999574$ & $0.999848$ \\
\hline
$0.3$ & $0.984961$ & $0.998289$ & $0.998834$ & $1$ &
$0.994985$ & $0.999297$ & $0.999572$ & $1$ \\
\hline
$0.4$ & $0.985188$ & $0.997929$ & $0.998585$ & $0.999799$ &
$0.995129$ & $0.999136$ & $0.999489$ & $0.999898$ \\
\hline
$0.5$ & $0.985921$ & $0.997789$ & $0.999233$ & $0.999839$ &
$0.995325$ & $0.999116$ & $0.999699$ & $0.999929$ \\
\hline
$0.6$ & $0.983814$ & $0.998475$ & $0.998730$ & $1$ &
$0.994607$ & $0.999352$ & $0.999530$ & $1$ \\
\hline
\end{tabular}
\end{center}
\caption{Adjusted Rand Index (ARI) and Normalized Mutual Information (NMI) of the optimal partition generated by modularity $\mathcal{Q}$ and null-adjusted persistence ${\mathcal{P}}^{\star}$ on the LFR graphs generated with $\tau_{1}=2$, $\tau_{2}=2$ and average degree (Av Deg) equal to $10$ and $15$.}
\label{Table1}
\end{table}



\subsection {Application to a real network}
We apply the persistence-based community detection method to a real-world network and compare our proposal with the modularity-based results. The scope is to assess how well the two functions can discover the ground-truth structure of the data, and to highlight the differences between the partitions they produce.
We refer to the social network described in \cite{Leskovec2012}. 
The network consists of the merge of ten ego-networks\footnote{An ego-network is a subgraph consisting of a focal node (ego), its directly connected neighbors (alters), and all edges between these alters.} of friendship relationships in Facebook, and contains $4039$ nodes and $88234$ connections. Each node represents an anonymized Facebook user from one of the ten friends' lists. Each edge corresponds to a friendship between two Facebook users. The network is undirected because edges in Facebook encode only reciprocal ties.\footnote{The dataset is available at this link: \href{https://snap.stanford.edu/data/ego-Facebook.html}{Stanford Facebook Dataset}.} The ego nodes are listed in Table \ref{Table2}. They typically have a high degree and aggregate their own friend lists around themselves, giving the network a coarse mesoscopic structure that could be used as a first level ground-truth community structure. However, the global network is much more complex than the mere aggregation of the ten ego-networks and, therefore, requires an adequate community search methodology.

Consistent with what was shown in Section \ref{sec3}, we expect that different choices of the objective function will produce different results in terms of community structure. Indeed, the two objective functions produce very different partitions: the Louvain algorithm that maximizes modularity produces a partition into $17$ communities ($M$-communities), while the algorithm that maximizes null-adjusted persistence finds a partition into $166$ communities ($P$-communities). To better compare the different findings, we refer to Table \ref{Table2}. The first and second columns list the ego nodes and their degree in the Facebook network. The third and fourth columns list the $17$ communities according to the Louvain algorithm for modularity ($M$-communities) and the number of nodes in each community. In the next column, we list the number of communities into which the null-adjusted persistence splits each of the $M$-communities. In the remaining three columns, we list the communities according to the null-adjusted persistence ($P$-communities) to which the ego node belongs, their size, and the percentage reduction of this size with respect to the $M$-community to which the same node belongs.
The last rows report the $M$-communities that do not contain any ego node.

\begin{table}[h]
	\begin{center}
		\begin{tabular}{|c|c||c|c||c||c|c||c|}
			\hline \hline
			Egonode	& Degree    & M-community & $\#$ Nodes  & Splitting  & P-community  & $\#$ Nodes  & Reduction  \tabularnewline \hline
			$\#1$               & $347$   &	1  & $341$ & $28$ &  $21$  & $176$  &  $48.39\%$  \tabularnewline \hline
			$\#349$          	& $229$   &	\multirow{2}{*}{2} & 	\multirow{2}{*}{$395$}  & \multirow{2}{*}{$17$}  &  $42$    &  $187$    & $52.66\%$ \tabularnewline \cline{1-2} \cline{6-8}
			$\#415$        	    & $159$   &	   &       &      &  $39$ & $61$    & $84.56\%$ \tabularnewline \hline
			$\#1685$  	        & $792$   &	3 & $561$ & $24$ &  $115$   & $199$    & $64.53\%$  \tabularnewline \hline
			$\#108$             & $1045$  &	4  & $428$ & $27$ & $69$  &  $347$   &  $18.93\%$ \tabularnewline \hline
			$\#1913$	        & $755$   &	5 & $423$ & $13$ & $98$   &  $278$   & $34.28\%$  \tabularnewline \hline
			-	                & -       &	6  & $25$  & $1$  & -  &  -  &  - \tabularnewline \hline
			$\#3981$	        & $59$    &	7 & $60$  & $7$  &  $160$  &  $35$   &  $41.67\%$ \tabularnewline \hline
			$\#687$          	& $170$   &	\multirow{2}{*}{8} & 	\multirow{2}{*}{$206$}  & 	\multirow{2}{*}{$11$}  & \multirow{2}{*}{$46$} & \multirow{2}{*}{133} &  \multirow{2}{*}{$35.44\%$} \tabularnewline \cline{1-2}
			$\#699$         	& $68$    &	   &       &      &    &     &    \tabularnewline \hline
			$\#3438$	        & $547$   &	9 & $548$ & $31$ &   $148$   &  $118$   & $78.47\%$ 	 \tabularnewline \hline
			-                   & -       &	10  & $386$ & $14$ & -  &  -  &  - \tabularnewline \hline
			-                   & -       &	11  & $54$  & $2$  & -  &  -  &  - \tabularnewline \hline
			
			-                   & -       &	12  & $38$  & $1$  & -  &     & - \tabularnewline \hline
			
			-	                & -       &	13  & $73$  & $1$  & -  &  -  &  - \tabularnewline \hline
			-                   & -       &	14 & $237$ & $3$  & -  &  -  & - \tabularnewline \hline
		
			-                   & -       &	15 & $19$ & $1$  & -  &  -  & - \tabularnewline \hline
			-                   & -       &	16 & $226$  & $5$  & -  &  -  & - \tabularnewline \hline
			-                   & -       &	17 & $19$  & $1$  & -  &  -  & - \tabularnewline \hline			
		\end{tabular}
	\end{center}
	\caption{Community structure of the Facebook network.}
	\label{Table2}
\end{table}
       
Facebook network and its partitions into $M$-communities and $P$-communities are shown in Fig. \ref{facebook}.

\begin{figure}[H]
	\centering
	\subfloat[]{\includegraphics[width=0.70\textwidth]{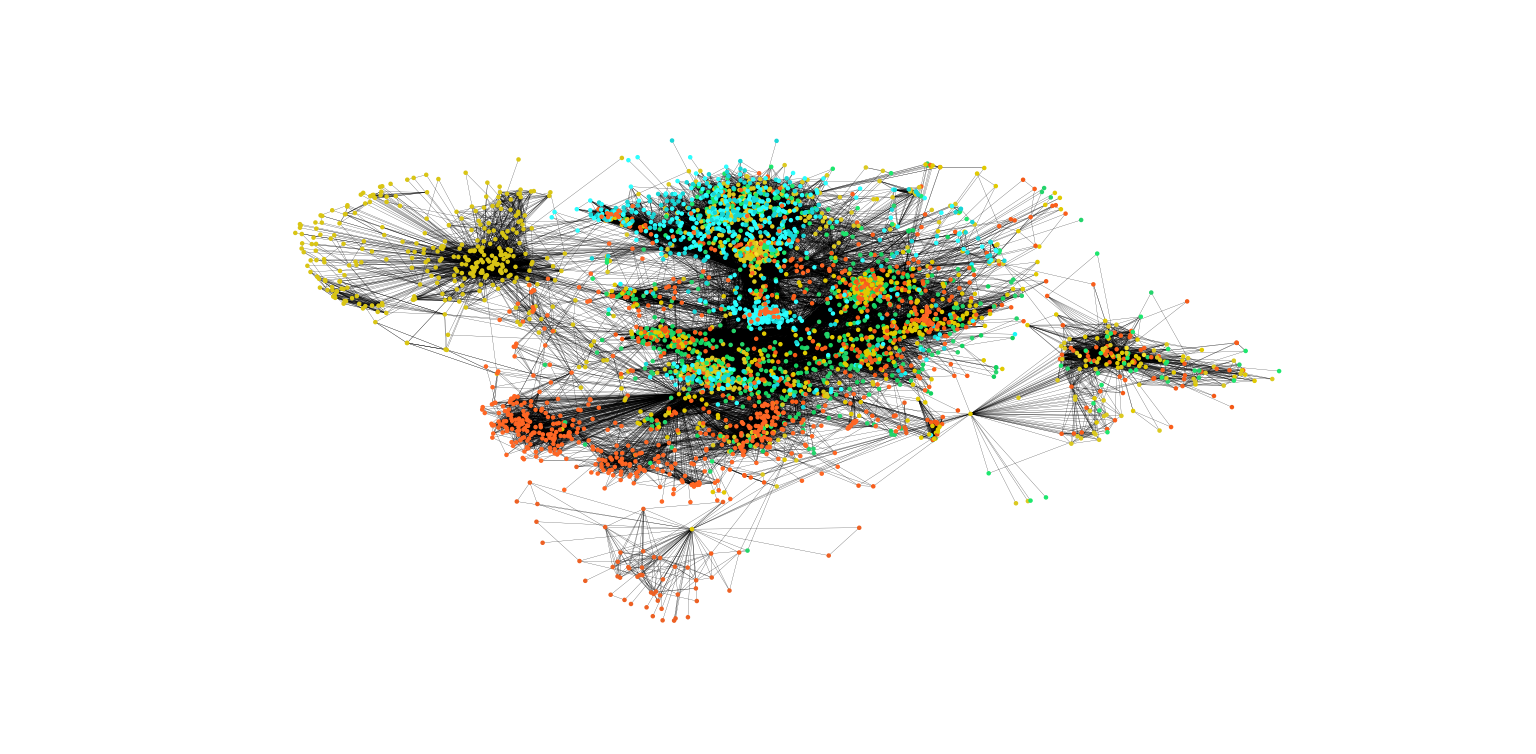}}\\
	\subfloat[]{\includegraphics[width=0.70\textwidth]{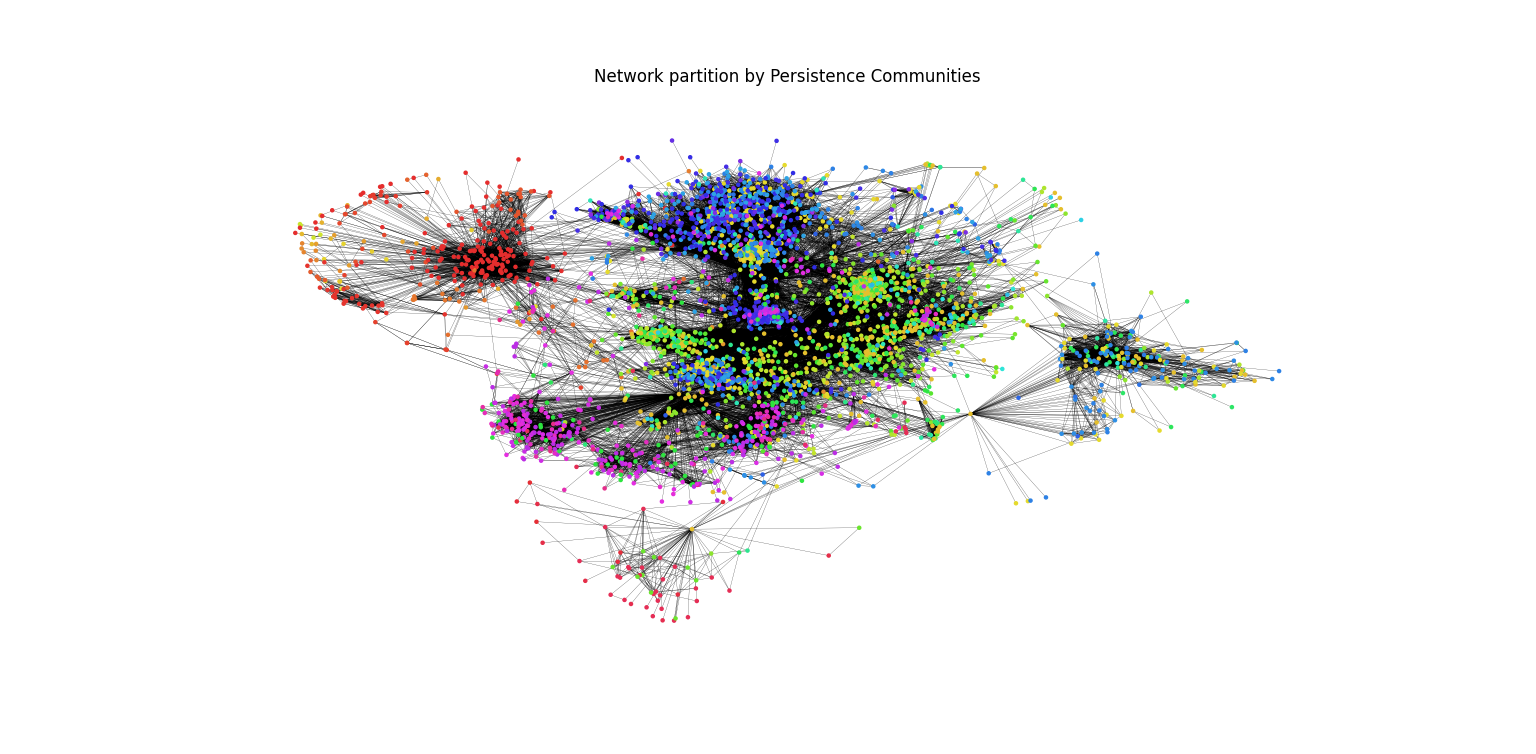}}
	\caption{Partition of the Facebook network into $M$-communities, panel (a), and into $P$-communities, panel (b).}
	\label{facebook} 
\end{figure}

The null-adjusted persistence returns $166$ communities with $2$ to $347$ nodes, with $61$ communities having more than $10$ nodes, $34$ communities having more than $20$ nodes, and $11$ communities having more than $100$ nodes ($P$-communities 21, 42, 46, 58, 69, 98, 105, 113, 115, 129 and 148).
By Table \ref{Table2}, we preliminarily observe that, in most cases, ego nodes tend to identify a self-community, and both objective functions are able to catch this. However, almost all $M$-communities split into a bunch of $P$-communities. Typically, $P$-communities realize an additional partition within the $M$-community. For example, $M$-community 9 is divided into 31 $P$-communities, the largest of which are four communities with 118, 82, 68 and 58 nodes, respectively, plus other communities with a few dozen nodes each. 
This is represented in Figure \ref{splittingcommunities}, panel (a): almost the totality of the communities found by the $M$-partition is fragmented into smaller communities by optimizing the null-adjusted persistence. 
The null-adjusted persistence catches the same nodes' relationships as the modularity does, but at a more refined level. Moreover, there are a few $P$-communities that contain nodes from different $M$-communities. They are only $14$ out of $166$, specifically:
13 (from 1 and 3),
39 (from 2 and 4),
64 (from 2, 4 and 10), 
69 (from 2, 4 and 10),
73 (from 2, 3, 9, 10 and 11),
77 (from 3, 4 and 10),
80 (from 4 and 10),
84 (from 2 and 4),
90 (from 4 and 10),
98 (from 5 and 14),
99 (from 14 and 15),
111 (from 3, 10 and 11),
115 (from 3 and 16),
165 (from 2 and 7).
These communities are represented in Fig. \ref{splittingcommunities}, panel (b).
The largest $P$-community is cluster $69$ with $347$ elements. It contains the ego node $\#108$ and collects nodes from $M$-communities $2$, $4$ and $10$. The other $P$-communities containing ego nodes, that is nodes $\#1,\ \#349,\ \#415,\ \#687,\ \#699,\ \#1685,\ \#1913,\ \#3438,\ \#3981$, are 21, 42, 39, 46, 46, 115, 98, 148 and 160, respectively.

It is worth noting that null-adjusted persistence not only breaks precisely those $M$ communities that contain ego nodes, but interestingly also leaves those that do not almost unchanged. Five of the $M$-communities that do not contain ego nodes do not undergo any change at all when analyzed by the null-adjusted persistence. For example, $M$-community 14 contains $19$ nodes with degree greater than $200$, but none of them is an ego node and null-adjusted persistence recognizes it almost identically in the $P$-community 105 with $231$ nodes (plus $4$ nodes in $P$-community 98 and $2$ nodes in $P$-community 99).\footnote{Consider that in the whole network there are 40 nodes with degree higher than $200$, so higher than the degree of four ego nodes.}
Conversely, communities that contain ego nodes, e.g., $M$-community 1 containing node $\#1$ or L-community 3 containing node $\#1685$, are broken. Peculiarly, they contain no other nodes with degree greater than $200$. Only node $\#108$, which is the most central node in the network in term of degree, seems to be able to attract other nodes, specifically 14 other nodes, with degree greater than $200$. 
Remarkably, persistence is able to break the two ego-networks of nodes $\#349$ and $\#415$, which are adjacent in the whole network.
Modularity cannot resolve these two ego-networks and merges them into a single community (the $M-$community 2), as they are quite deeply nested into the global network. This does not happen for the two most peripheral nodes $\#687$ and $\#699$, which are also adjacent nodes but are assigned to a single community by both modularity and null-adjusted persistence. Therefore, $P$-communities seem to provide a more refined representation of the relationships' structure on this network, by preserving in any case the central role of the ego nodes. Null-adjusted persistence is better than modularity at recognizing nested communities that have independent origins as distinct ego-networks.
\begin{figure}[H]
	\centering
        \subfloat[]{\includegraphics[width=0.45\textwidth]{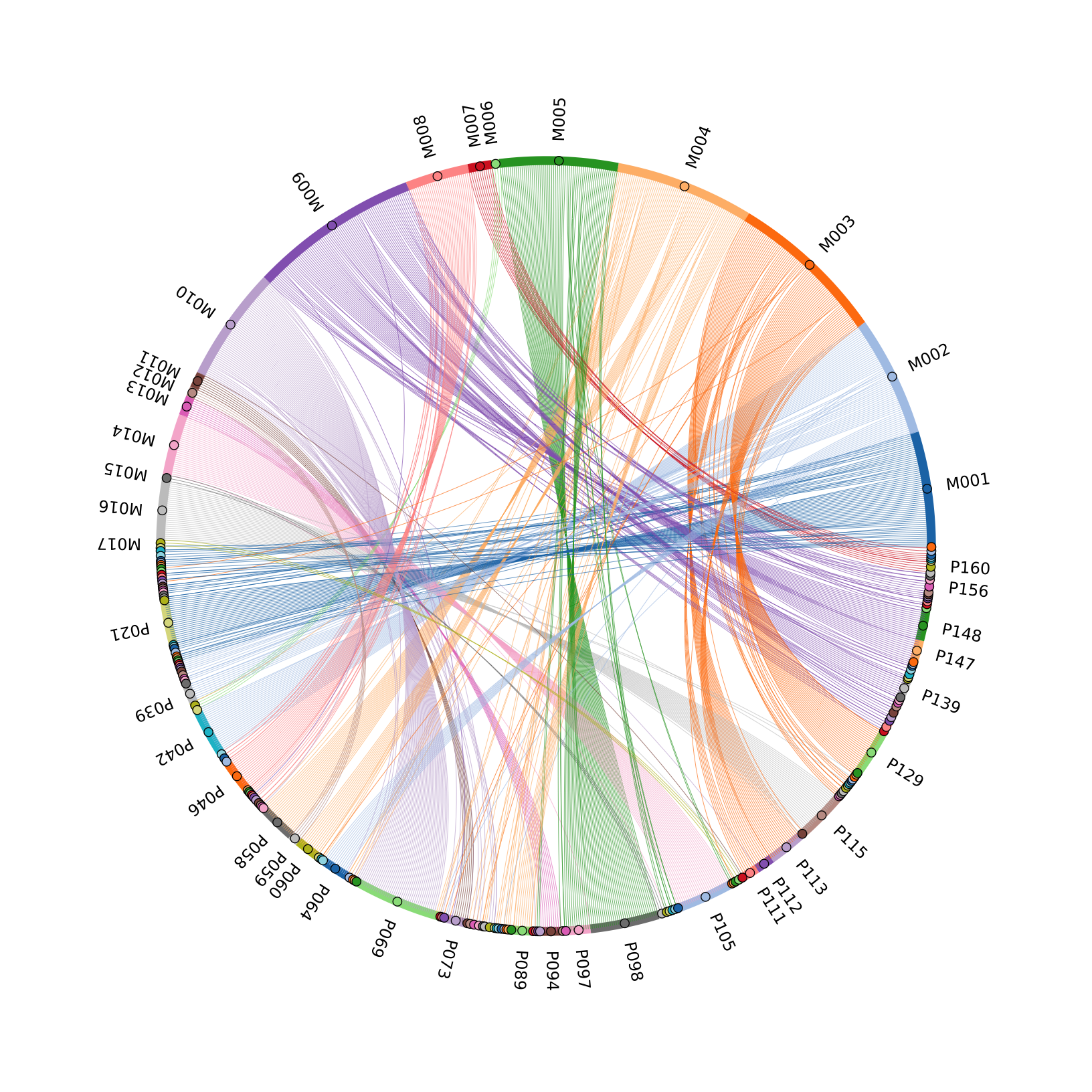}\label{splittingcommunitiesv1}}\\
	\subfloat[]{\includegraphics[width=0.45\textwidth]     {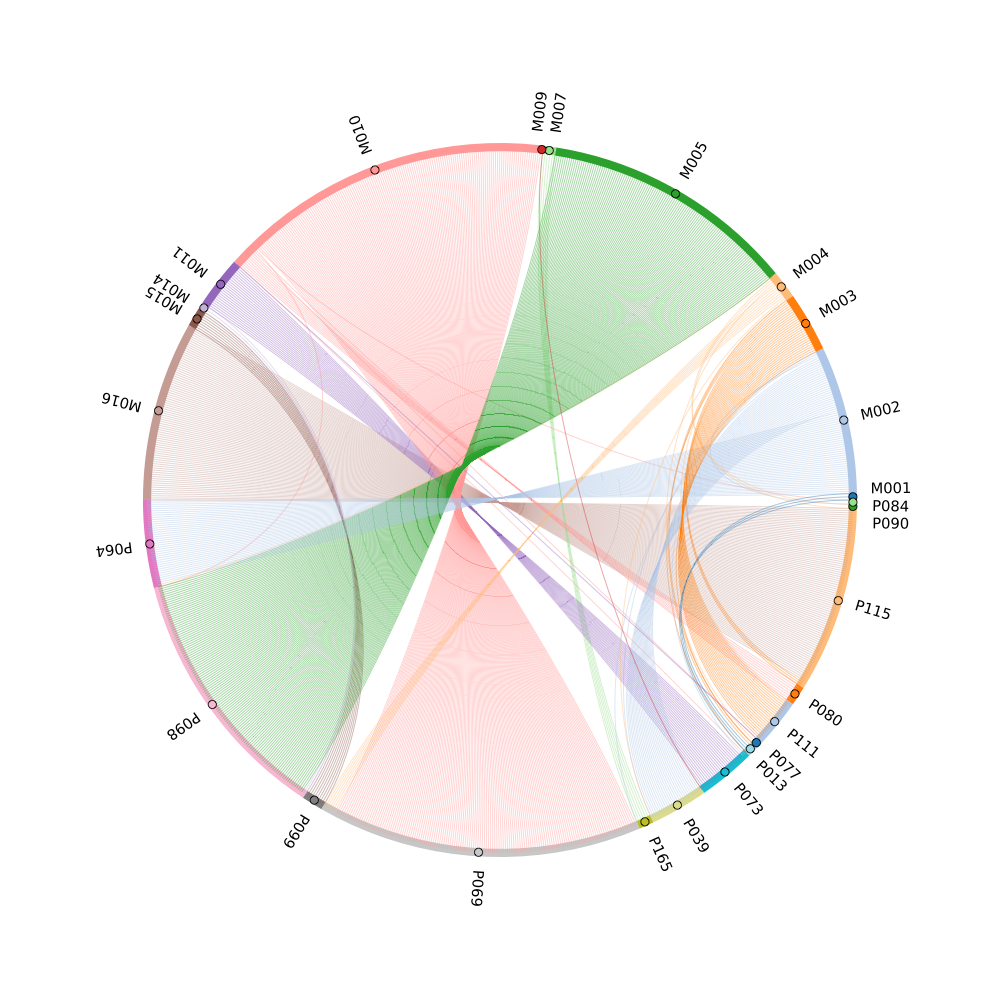}\label{splittingcommunitiesv2}}

	\caption{Splittig M-communities into P-communities}
	\label{splittingcommunities} 
\end{figure}

\section{Conclusion}
This paper relates the concept of persistence probability to the community structure of a network. We introduce the null-adjusted persistence, which implies a comparison of persistence probability with a null model, and we adopt it as the objective function of the optimization problem for community detection. This proposal incorporates features from persistence probability and modularity, offering significant advantages over both these functions.
Indeed, it aligns with modularity-based approaches by separating observed persistence into contributions from the null model and deviations from it, thereby improving interpretability and integration with other network analysis methods. Its ability to take both positive and negative values allows to distinguish cohesive clusters from those that are less cohesive than expected. This feature is particularly valuable in networks with strong degree heterogeneity, where this heterogeneity could otherwise produce misleading results. It also allows to detect partitions with a higher resolution than modularity in large networks consisting of many medium-small communities. 

A first direction of future research is to develop more refined and optimized heuristics to search for such communities. Once developed suitable algorithms, it will be possible to test the potential of null-adjusted persistence on large real networks and highlight its usefulness in improving the robustness and applicability of community detection methods.

\appendix

\section{Mixed Integer Linear Programming Formulation of the Optimal Persistence Partition}
\label{appendixA}
In the following, we report the formulation of problem \eqref{maxpersistence} through a mixed integer linear programming. Note that from proposition (\ref{proposition1}) for any partition, the difference between the the null-adjusted persistence and the persistence probability is a constant, therefore, any maximizer of the latter is a maximizer of the former too.
Suppose that $V$ can be partitioned into $k = 1, \ldots, n$ slots (slots represent clusters, some of them possibly empty for numerical convenience), then let the decision variable $z_{ik} = 1$ if node $i$ is assigned to slot $k$, 0 otherwise. The following objective function represents the persistence of a node-to-slot assignment that has to be maximized under the set of binary variables $z$:
\begin{equation}
\label{milp:fo}
{\max}_{z} \sum_{k=1}^n \frac{\sum_{i = 1}^{n-1}\sum_{j = i+1}^n 2a_{ij}(z_{ik}z_{jk})}{\sum_{i = 1}^{n-1}\sum_{j = i+1}^n [a_{ij}(z_{ik}z_{jk}) + a_{ij}\max\{z_{ik},z_{jk}\}]}
\end{equation}

In the objective function, the product $(z_{ik}z_{jk})$ is 1 if and only if both $i$ and $j$ are assigned to the same slot $k$, and therefore arc $(i,j)$ is internal to the cluster represented by the slot $k$. The term $\max\{z_{ik},z_{jk}\}$ is 1 if at least one between $i$ and $j$ is internal to the slot $k$, and therefore either arc $(i,j)$ is internal to the slot $k$, or is in the cut from the slot $k$ to another slot. The numerator expresses the number of arcs internal to a cluster, counted twice; the denominator expresses the number of arcs internal to a cluster, plus the number of the arcs exiting the cluster. The ratios of the objective function are written with the convention $0/0 = 0$. 

A partition $\Pi$ can be represented by many node-to-slot assignments, simply relabeling the slot containing a given cluster. Therefore some constraints must be imposed to avoid multiple symmetric solutions (as they would increase the computational times exponentially). These constraints are:
\begin{equation}
\label{milp:c1}
\begin{split}
& \sum_{k = i}^n z_{ik} = 1\, \text{ for all }i \in V\\
& z_{ik} \le z_{kk} \, \text{ for all }i \in V, k > i.
\end{split}
\end{equation}

The first kind of constraint requires that every node must be assigned to one slot, and that the index of the slot must be greater than or equal to the node index. The second kind of constraint imposes that node $i$ can be assigned to a slot $k$ only if $k$ is not empty and node $k$ has been assigned to slot $k$. Taken together, the two constraints impose that for a cluster ${\mathcal{C}}_\alpha = \{i_1,\ldots, i_r\}$, 
the bin that contains ${\mathcal{C}}_\alpha$ can only be $k = i_r$. 
These constraints are important for numerical purposes to avoid symmetric solutions in branch\&bound. They were previously used in \cite{Benati2017, Temprano2024}.
Note that the optimal solution of \eqref{milp:fo} is made up of connected clusters. Indeed, by absurd, if the optimal solution contained at least one bin representing a unconnected cluster, then it could be split into at least two connected components, improving the objective function, and thus contradicting its optimality.

The problem is not linear. However, it can be turn into a mixed integer linear programming with the appropriate constraints and linearization of the quadratic terms.
The product terms of the objective function can be linearized as:
\begin{equation*}
x_{ijk} = z_{ik} z_{jk} \Longleftrightarrow 
	\begin{cases}
		x_{ijk} \le z_{ik} \\
		x_{ijk} \le z_{jk} \\
		x_{ijk} \ge z_{ik} + z_{jk} - 1
		\end{cases}
		\hspace{3mm} \forall i,j\in V.
\end{equation*}
The $\max$ term of the objective function can be linearized as:	
\begin{equation*}     
	y_{ijk} = \max\{z_{ik}, z_{jk}\} \Longleftrightarrow
	\begin{cases}
		y_{ijk} \ge z_{ik} \\
		y_{ijk} \ge z_{jk} \\
		y_{ijk} \le z_{ik} + z_{jk}
	\end{cases}
		\hspace{3mm} \forall i,j\in V.
\end{equation*}
The objective function now reads:
\begin{equation*}  
{\max}_{x,y} \sum_{k=1}^n \frac{\sum_{i = 1}^{n-1}\sum_{j = i+1}^n 2a_{ij} x_{ijk}}{\sum_{i = 1}^{n-1}\sum_{j = i+1}^n [a_{ij}x_{ijk} + a_{ij}y_{ijk}]}.
\end{equation*}
Note that the objective function is increasing with respect to variables $x_{ijk}$ and decreasing with respect to $y_{ijk}$, therefore some of the above linearization terms are not necessary to characterize the optimal solution. 
	
Now the objective function is the sum of ratios between two linear terms, then it can be linearized using the Charnes-Cooper linearization. Introduce a new variable $u_k$, defined as:
\begin{equation*}  
u_k = \frac{1}{\sum_{i = 1}^{n-1}\sum_{j = i+1}^n [a_{ij}x_{ijk} + a_{ij}y_{ijk}]}
\end{equation*}
so we obtain the constraint:
\begin{equation*}  
\sum_{i = 1}^{n-1}\sum_{j = i+1}^n [a_{ij}x_{ijk}u_k + a_{ij}y_{ijk}u_k] = 1.
\end{equation*}

The constraint above is necessary only on the condition that bin $k$ is non-empty, otherwise it must be 0 to respect the convention $0/0 = 0$. Considering that, from the anti-symmetric constraints, a slot $k$ is non empty if and only if $z_{kk} = 1$, we have that the above condition is extended to all $k$ with:
\begin{equation*}  
\sum_{i = 1}^{n-1}\sum_{j = i+1}^n [a_{ij}x_{ijk}u_k + a_{ij}y_{ijk}u_k] = z_{kk} \, \mbox{ for all $k$}.
\end{equation*}
Quadratic terms can be linearized:
\begin{equation*}         
	x_{ijk}u_k = w_{ijk} \Longleftrightarrow
	\begin{cases}
		w_{ijk} \le u_k \\
		w_{ijk} \le x_{ijk} \\
		w_{ijk} \ge u_k - (1-x_{ijk})
	\end{cases}
		\hspace{3mm} \forall i,j\in V.
\end{equation*}
and similarly:
\begin{equation*}        
	y_{ijk}u_k = q_{ijk} \Longleftrightarrow
	\begin{cases}
		q_{ijk} \le u_k \\
		q_{ijk} \le y_{ijk} \\
		q_{ijk} \ge u_k - (1-y_{ijk})
	\end{cases}
		\hspace{3mm} \forall i,j\in V.
\end{equation*}
and now the objective function is:
\begin{equation}\label{milp:fo1}
{\max} \sum_{k=1}^n \sum_{i = 1}^{n-1}\sum_{j = i+1}^n 2a_{ij} w_{ijk}.
\end{equation}

Combining the objective function \ref{milp:fo1} with the constraints representing the quadratic terms and the antisymmetric conditions, we obtain a MILP that can be solved by any Integer Linear Programming solver for instances of moderate size, such as the networks arising from opinion surveys, see \cite{benati2024}. 

In the following test, we used Gurobi 11.02 in the R/RStudio environment. First, we simulate a Caveman graph, see \cite{Watts1999}, made up of cliques of five nodes. Then, cliques are multiplied by 3, 4 and 5, to obtain graphs of 15, 20 and 25 nodes. Caveman graphs have a clear community structure, making them the easiest to solve because there are not many competing solutions. Next, caveman graphs are randomly rewired to hide the community structure and make instances harder to solve.   

The computational times are reported in Table \ref{tab:MILP}, expressed in seconds. Assuming a time limit of 3,600 seconds, we observe that instances of 20 nodes can be solved in a few minutes if the graph has a recognizable community structure. However, the time required increases significantly as the number of nodes increases. Instances of 25 nodes are solved within an hour, but after that size, that is between 26 and 30 nodes, the time limit is often exceeded. It is worth noting that the optimization problem \ref{milp:fo} is similar to the min-cut density clustering analyzed in \cite{Temprano2024}: both optimization functions are the sum of ratios, but the direction of the optimization is reversed, maximization vs minimization. The computational times of the two models are similar, but in \cite{Temprano2024} it has been shown that column generation with branch\&cut can improve the computational times, and therefore it is an interesting direction of new research.
\begin{table}[h]
    \centering
    \begin{tabular}{|c|c|c|c|c|}
        \hline
        nodes & Caveman & \multicolumn{3}{c|}{Caveman rewired} \\ \hline
              &        & min & average & max  \\ \hline
        15    & 4.9   & 6.5  & 9.3   & 13.6 \\ \hline
        20    & 57.5   & 63.6 & 212.4   & 1140.3 \\ \hline
        25    & 1069.1   & 1450.2      & 2022.8  & 2952.1 \\ \hline
    \end{tabular}
    \caption{Computational times (seconds) of the Mixed Integer Linear Programming.}
    \label{tab:MILP}
\end{table}

\section{Louvain-based algorithm}
\label{appendixB}

We describe in detail the algorithm used for simulations on synthetic and real-world networks in Section \ref{sec4}.  The algorithm falls into the category of standard "greedy" optimization algorithms, and it follows an approach similar to that of the Louvain method.

\begin{algorithm}[h]
    \caption{Milan: A Louvain-based algorithm for persistence}\label{alg:louvain:persistence}
    \DontPrintSemicolon
    \SetKwFunction{Move}{Move}
    \SetKwData{Left}{${\Pi}^{\star}$}
    \KwIn{a network $G=(V,E)$.}
    \KwResult{A partition ${\Pi}^{\star}$ of $V$ in communities.}
    \BlankLine

    ${\Pi}^{\star}\gets\{\{1\},\ldots,\{\lvert V\rvert\}\}$\;\label{alg:start}
    \While{\texttt{True}} {
        ${\Pi}^{\prime}\gets \Move{G, \Left}$ \;
        \If{${\Pi}^{\prime}={\Pi}^{\star}$}{
            \Break \;
        }
        ${\Pi}^{\star}\gets {\Pi}^{\prime}$\;
    }

    \setcounter{AlgoLine}{0}
    \BlankLine
    \BlankLine
    \SetKwProg{Fn}{}{}{}
    \SetKwFunction{Move}{Move}   
    \Fn(){\Move{$G$, $\{\mathcal{C}_{1}, \ldots, \mathcal{C}_{q}\}$}}
    {
        \KwIn{$G=(V,E)$ and $\{\mathcal{C}_{1}, \ldots, \mathcal{C}_{q}\}$ a partition of $G$ in communities.}
        \KwResult{a $q$-connected subset of $V$.}
        \BlankLine
        $\Pi\gets \{\{\mathcal{C}_{1}\}, \ldots, \{\mathcal{C}_{q}\}\}$ \;\label{move:start}
        \While{\texttt{True}}{
            \ForEach{$\mathcal{C} \in \{\mathcal{C}_{1}, \ldots, \mathcal{C}_{q}\}$} 
            {
                let $\mathcal{C}^{\prime}\in \Pi$ the  community with largest increase in null-adjusted persistence $\Delta \mathcal{P}^*$ when $\mathcal{C}$ and $\mathcal{C}^{\prime}$ are merged. \; \label{move:best:delta}
                
                \If{$\Delta \mathcal{P}^* > 0$ }
                {
                    update $\Pi$ by merging $\mathcal{C}$ and $\mathcal{C}^{'}$.\;
                }
            }
            \If{$\Pi$ has not been updated}{
                \Break \;
            }
        }
        \Return $\Pi$
    }
\end{algorithm}

The algorithm steps are described in Algorithm \ref{alg:louvain:persistence}. Starting with each vertex as the unique member of a community (Line~\ref{alg:start}), the algorithm repeatedly calls the function \Move which modifies the initial partition to improve the null-adjusted persistence.
The algorithm stops when a call of the function \Move does not change the input partition, i.e. no more improvement is found.

The function \Move is the core of the algorithm. At the beginning, communities are considered as individual nodes that are the only members of an initial partition $\Pi$ (Function \Move Line~\ref{move:start}).
The function repeatedly modifies the partition $\Pi$ by merging two of its communities $\mathcal{C}^{\prime}$ and  $\mathcal{C}$ if this merge produces a gain in the objective function.

Specifically, at each step, the selected community pair $(\mathcal{C}^{\prime}, \mathcal{C})$ is the one that results in the greatest positive increase in the null-adjusted persistence $\Delta \mathcal{P}^*$ (Function \Move Line~\ref{move:best:delta}).
Then, $\Pi$ is modified accordingly, and ${\mathcal P}^{\star}_{\Pi}$ increases. Otherwise, the function \Move stops, returning the current partition $\Pi$.
Since joining a pair of communities without common edges can never increase null-adjusted persistence, the algorithm only considers pairs of communities connected by at least one edge; therefore, the number of pairs of communities is approximately the number of edges $|E|$ in $G$.



\textbf{Data availability}

Data will be made available on request.

\textbf{Acknowledgments}
AA, PB, SB and RG acknowledge financial support from the European Union – NextGenerationEU. Project PRIN 2022 “Networks:decomposition, clustering and community detection'' code: 2022NAZ0365 - CUP H53D23002510006. PB and RG are members of the GNAMPA-INdAM group. CC has received funding from the European Union’s Horizon 2020 research and innovation program under the Marie Sklodowska-Curie Grant Agreement No 101034403.

\bibliographystyle{model5-names}\biboptions{authoryear}

\bibliography{references}

\end{document}